\documentclass[a4paper,11pt]{article}
\usepackage{fullpage}
\usepackage{times}
\usepackage{subfigure}
\usepackage{vmargin}
\usepackage{fancyheadings}
\usepackage[cp1251]{inputenc}
\usepackage[english]{babel}
\usepackage{graphicx}
\usepackage{color}
\usepackage{cite}
\usepackage{amsmath,amssymb,amsthm,enumerate,bbm,color,verbatim,setspace}
\newcommand{\tr}{\text{Tr}}
\usepackage{braket}
\theoremstyle{plain}  
\newtheorem{thm}{Theorem}[section]

\theoremstyle{remark}

\date{}
\title{HERMITIAN AND UNITARY ALMOST-COMPANION MATRICES OF POLYNOMIALS ON DEMAND}
\author{L.A. Markovich$^{1,2,3,4*}$, A. Migliore$^{5}$ and A. Messina$^{6}$\\
      $^1$Instituut-Lorentz, Universiteit Leiden, P.O. Box 9506,\\ 2300 RA Leiden, The Netherlands\\
 $^2$QuTech and Kavli Institute of Nanoscience, Delft University of Technology,\\ 2628 CJ, Delft, The Netherlands\\
$^3$Institute for information transmission problems, Moscow,\\ Bolshoy Karetny per. 19, build. 1, Moscow 127051, Russia\\
$^4$Russian Quantum Center, Skolkovo, Moscow 143025, Russia\\
$^5$Department of Chemical Sciences, University of Padova,\\ Via Marzolo 1, 35131 Padova, Italy\\
$^6$Dipartimento di Matematica ed Informatica dell'Universita di Palermo,\\ Via Archirafi 34, 90123 Palermo, Italy\\
$^*$Corresponding author e-mail: markovich@mail.lorentz.leidenuniv.nl}
\date{}
\begin{document}
\maketitle
\pagenumbering{arabic}
\begin{abstract}\noindent

We introduce the concept of Almost-Companion Matrix (ACM) by relaxing the non-derogatory property of the standard Companion Matrix (CM). That is, we define an ACM  as a matrix whose characteristic polynomial coincides with a given monic and generally complex polynomial. The greater flexibility inherent in the ACM concept, compared to CM, allows the construction of ACMs that have convenient matrix structures satisfying desired additional conditions, compatibly with specific properties of the polynomial coefficients. We demonstrate the construction of Hermitian and Unitary ACMs starting from appropriate third degree polynomials, with implications for their use in physical-mathematical problems such as the parameterization of the Hamiltonian, density, or evolution matrix of a qutrit. We show that the ACM provides a means of identifying properties of a given polynomial and finding its roots. For example, we describe the ACM-based solution of cubic complex algebraic equations without resorting to the use of the Cardano-Dal Ferro formulas. We also show the necessary and sufficient conditions on the coefficients of a polynomial for it to represent the characteristic polynomial of a unitary ACM. The presented approach can be generalized to complex polynomials of higher degree.
\end{abstract}
\medskip


\section{\label{sec:level1}Introduction}
\par 
Given a complex and monic\footnote{A monic polynomial is a univariate polynomial in which the leading coefficient  is equal to 1.} polynomial $P_n(z)$, it is always possible to define a matrix with a specified arrangement of the polynomial coefficients as its entries, such that $P_n(z)$ coincides with  the characteristic polynomial of the matrix. The set $S(P_n(z))$ of all these $n\times n$ matrices with complex entries, sharing the same characteristic  polynomial $P_n(z)$, is infinite and can include both derogatory and nonderogatory\footnote{A square matrix is called derogatory (nonderogatory) when the degree of its minimal polynomial is less than (equal to) the order $n$ of the matrix. The term derogatory has been coined by Sylvester in the early years following 1880. Etymologically, it probably originates from the Latin verb \textit{"derogare"}, in its particular meaning of "decrease", to underline that the degree of the minimal polynomial of the matrix is less than $n$.} matrices.
We observe, in fact, that, by definition, the \textit{Frobenius Matrix} \cite{Frobenius} of the shared characteristic polynomial always belongs to $S(P_n(z))$.  
A remarkable property of this matrix, which stems directly from its construction, is the coincidence between its characteristic and minimal polynomials, whatever $P_n(z)$. 
Therefore this  matrix is classified as nonderogatory and, following  Horn and Johnson \cite{Horn}, it is known as the \textit{Companion Matrix} (CM)\footnote{As reported by Hawkins \cite{Hawkins}, the term \textit{companion matrix} was coined by Loewy in 1917 \cite{Loewy_1917}. In 1946 MacDuffee introduced the term "companion matrix" as a translation from the German "Begleitmatrix".} of its characteristic or minimal polynomial.
\par  
Henceforth, we refer to the Frobenius Matrix as the Frobenius Companion matrix ($FCM$) of $P_n(z)$.
When the algebraic multiplicity of each of the $n$ eigenvalues of the $FCM$ is 1, any matrix in $S(P_n(z))$ is nonderogatory, since these $n$ distinct eigenvalues are all necessarily roots of its  minimal polynomial, which therefore has degree $n$ \cite{Horn}.
We also remark that  this condition guarantees that each matrix in the $S(P_n(z))$ set is diagonalizable \cite{Horn} and that, consequently, all matrices $\in S(P_n(z))$ can be generated from the $FCM$ $\in S(P_n(z))$ by means of a similarity transformation. In this way, by definition and under the conditions established for the spectrum $\sigma(FCM)$ of $FMC$, $S(P_n(z))$ includes all and only the companion matrices of $P_n(z)$. 

When, instead, the distinct roots of the common characteristic polynomial are $p< n$, $S(P_n(z))$ includes infinite derogatory matrices, which cannot be structurally similar to the Frobenius matrix\footnote{We point out that, in this case, the Frobenius matrix is no longer diagonalizable, since, being nonderogatory by construction, it necessarily has at least one eigenvalue with geometric multiplicity smaller than the (algebraic) multiplicity.} and to any nonderogatory matrix belonging to $S(P_n(z))$. For example, consider that the set $S(P_n(z))$ contains the diagonal matrices whose $n$ entries are nothing but the $p$ distinct roots of the characteristic polynomial repeated as many times as their multiplicities, whose sum is $n$. The degree of their characteristic polynomial is $n$, while the degree of their minimal polynomial is $p<n$ \cite{Horn}. Therefore, these matrices are derogatory, as are the infinite matrices generated from them by similarity.

The Frobenius matrix dates back to 1879 and was given in the form \cite{Frobenius}:
 \begin{eqnarray}C_n=\left(
                        \begin{array}{ccccc}
                          0 & 0 & \ldots & 0 & -c_{n-1} \\
                          1 & 0 & \ldots & 0 & -c_{n-2} \\
                          0 & 1 & \ldots & 0 & -c_{n-3} \\
                          \vdots & \vdots & \ddots & \vdots & \vdots \\
                          0 & 0 & \ldots & 1 & -c_{1} \\
                        \end{array}
                      \right)\label{2_1}
\end{eqnarray}
by the German mathematician Ferdinand Georg Frobenius \cite{Hawkins}. Sometimes, in the literature it is presented in three other unitarily transformed forms, all parameterized in terms of the $n$ coefficients of the polynomial $P_n(z)$ and with the same number of entries equal to 0 or 1 as in \eqref{2_1} \cite{Barnett_2}. When the $n$ eigenvalues of the Frobenius matrix are distinct, the unitary matrix that yields the diagonal form of \eqref{2_1} is the Vandermont matrix of its $n$ eigenvalues. This property, as well as other properties and some applications of the Frobenius matrix may be found in \cite{Barnett_2}. 
\par Despite the fact that the first CM  was proposed more than 140 years ago, the  generation of different CMs (with further properties) has attracted a great deal of application research activity. CMs emerge naturally in mathematical methods for finding and characterizing the  roots of polynomials \cite{doi:10.1137/120865392,DETERAN2014197,bini1996numerical,HIGHAM20035,bueno2011recovery,antoniou2004new} and can be applied as well in the determination
of solutions of high-order scalar linear differential and difference equations \cite{brand1968applications}. CMs 
can give matricial representations of some fields \cite{wardlaw1994matrix} and are widely used in control theory, for example in writing the controllable canonical form associated with the transfer function of a system \cite{szederkenyi2006intelligent}. 
 The product of CMs is also used in the study of random walks and Markov chains \cite{LIM20112921}.
\par The original structure of the FCM shows a modest level of flexibility and, in fact, has stimulated the search and the emergence of generalizations leading to new proposals of CMs that pave the way to applications beyond the FCM. To this end, a successful strategy is based on a radical change of the basis in which the characteristic  polynomial $P_n(z)$ of the FCM is represented.
By definition, $P_n(z)$ is monic and is written as the sum of $z^n$ and a linear combination with coefficients in $\mathbb{C}$ of $1,z,z^2,\dots, z^{n-1}$. All these $n+1$ powers of $z$  constitute the  monomial basis which can be replaced by other polynomial bases. In this way, one can introduce new (still nonderogatory) CMs with nonvanishing elements on the main, sub and super diagonals.   
For example, the Chebyshev basis, independently adopted  by Specht \cite{Specht_3,Specht_4} and Good \cite{Good_1961}, led to a new CM called the \textit{colleague matrix}. This approach has been further generalized by Barnett \cite{Barnett_2}, who considered a basis of orthogonal polynomials and named the newly emerged CMs \textit{comrade matrices}. Subsequently, Barnett proposed to call \textit{confederate matrices} the CMs arising from the use of a general polynomial basis \cite{maroulas}.
\par The applications dedicated to the classic problem of finding the real zeros of a real coefficient polynomial of arbitrary degree deserves a special mention, because in this context the CMs have inspired a different approach, alongside exquisitely mathematical and computational investigations \cite{Branden,HAGLUND20001017,PITMAN1997279,WAGNER1991138,WAGNER1992459,WANG200563}. In the last decade, new quantum theory-based root-finding algorithms exploiting the construction of Hermitian companion matrices \cite{7016940,Nagata2018,Nagata2019,Tansuwannon2019,TAN200275,Weigert_2003} have also been proposed, thus increasing the interest in the study of CMs and related matrices in the rapidly growing research field of quantum computing.
\par The above is the general context for the main question addressed by this paper: is it possible to find a CM of a real or complex polynomial that is also hermitian/unitary, for example, or possesses some other prescribed special matrix structure? 
 This question has so far been answered only partially. 
 \par In \cite{SCHMEISSER199311} the CM for a $P_n(z)$ with real coefficients and real zeros is  constructed  as a real symmetric tridiagonal Hermitian matrix. This provides a complete solution to a problem raised and partly solved by M. Fiedler in \cite{FIEDLER1990265}, which has rekindled the interest in the general structure of CMs. Note that the Frobenius matrix is itself a Fiedler matrix after a reverse permutation matrix similarity. 
In \cite{EASTMAN2014255} all CMs are characterized in terms of combinatorial structure to generate new CMs. It is interesting to note that both the Frobenius and Fiedler CMs are sparse matrices, as they have $2n-1$ nonzero elements \cite{deaett2019non}. A new class of sparse CMs, also known as intercyclic CMs, was introduced in \cite{EASTMAN2014255} and  includes the Fiedler matrices as a special case.
 In \cite{deaett2019non} the non-sparse CMs are introduced noting that they are not connected with the sparse ones by a reverse permutation matrix similarity.
\par  To the best of our knowledge, the question of whether the CM of a complex polynomial can be sought as unitary or Hermitian is still open.

\subsection{Purpose and contribution of this study}
\par  CMs generally show relatively limited versatility due to their combinatorial structure.  For example, the FCM is never Hermitian or unitary. Searching for CMs that satisfy additional constrains of this kind is important when a given class of parametric polynomials is designed to reach a reliable theoretical control of a quantum  physical scenario.
Finding exact flexible solutions of well-defined inverse problems of this kind is a target of the present study. 
We stress that, for application purposes, we often do not need to associate nonderogatory matrices to a given polynomial. 
To emphasize this peculiar aspect of our matrix construction, we introduce the term \textit{Almost-Companion Matrix} (ACM) to refer to matrices that have a given polynomial as their characteristic polynomial but can be derogatory. Clearly, every CM is also an ACM, that is, the set of all ACMs is a superset of the set of all CMs of a given polynomial.
\par In short, in this study we address the following inverse problem: given a real or complex polynomial of any kind (for example, it may belong to a special class of polynomials, which is reflected in some special condition satisfied by its coefficients), we find a parametric ACM\footnote{By parametric we mean that the entries of the ACM can be functions of the polynomial coefficients.} that satisfies preassigned conditions (e.g., those required to be Hermitian or unitary) and whose characteristic polynomial coincides with the given one. For definiteness, our investigation is here limited to complex polynomials of third degree ($n= 3$), and its extension to higher degree polynomials is discussed. 
\par  The solution of the inverse problem outlined above for polynomials of order $n= 3$ is accompanied by some useful applications.
The relaxed constraints that characterize an ACM, as compared with a CM, allow the freedom to search, from the outset, for a trial ACM of a given $P_n(z)$  that is a Hermitian, unitary, or positive matrix, for example. If our inverse problem can be solved systematically through an approach that finds such ACMs whatever the given $P_n(z)$, then we readily have at our disposal a good platform for successful applications to problems such as those mentioned below.

\subsection{Physical applications}
Constructing an ACM of a given generally complex polynomial, in addition to being interesting in itself, also has considerable applications. In elementary algebra, for example, it could be a solution tool for counting the number of the real roots (and consequently that of the complex roots in the case of a complex polynomial) of a real or complex polynomial. In addition, it may help determine a (or the only) real root of a real polynomial of odd degree. In quantum mechanics or quantum information, it could provide new parametric representations of the density operator or the evolution operator of a physical system living in a finite dimensional Hilbert space.
\par The results found for a generic complex polynomial can be applied to the important particular case of a real polynomial. Investigating such a link is certainly of interest in Physics. For example, in classical physics, cubic real polynomials appear when looking for the principal axes of symmetric Cartesian tensors of rank two, such as inertia or magnetic/electric dipolar tensors \cite{borisenko1968vector}. In quantum mechanics they enter the scene as characteristic polynomials of any observable of a physical system that lives in a three-dimensional Hilbert space, such as, for example, a three-level atom or a qutrit.
Recipes for constructing an ACM of a cubic polynomial possessing real roots after the appropriate assignment of parametric real coefficients could provide an easy way to build, e.g.,  Hamiltonian qutrit models on demand for control purposes, or even the density operator describing a mixed state of a three-level atom.
\par The paper is organized as follows. In Sec.~\ref{inverse problem} we formulate the inverse problem consisting in the search of the ACM for a generic complex polynomial. In Sec.~\ref{sec:cubic} we construct the ACMs for a generic cubic complex polynomial.
Through these ACMs, we introduce a way to find the roots of the given polynomial without  using the Cardano-Dal Ferro formulas.
The case of the polynomial with real coefficients is also discussed in detail. In Sec.~\ref{densitymatrix} we present an application in quantum mechanics, constructing on demand the density matrix of a qutrit system  as an ACM. 
Sec.~\ref{unitarymatrix} shows the construction on demand of the unitary ACM of a qutrit. 
Possible extensions to higher degree polynomials, as well as further possible applications, are discussed in Sec.~\ref{discussion_conclusion}.

\section{Formulation of the Inverse Problem}
Consider a  matrix $A\in M_n$, where $M_n$ is the set of all $n\times n$ matrices over the complex field $\mathbb{C}$. Denoting  $I_n\in M_n$ the identity matrix, the monic polynomial in the complex variable $z$ 
\begin{eqnarray}\det{(zI_n-A)}&=&z^n +c_1z^{n-1}+\cdots c_{n-1}z +c_{n}\label{n0}\end{eqnarray}
is, by definition,  the characteristic polynomial of $A$ and belongs to the set $\mathbb{P}{_n}[\mathbb{C}]$ of all complex monic polynomials of degree $n$.

 It is well known that its $n$ coefficients $c_k$,  $k=1,2\dots,n$ contain information about the elements of $A$ that is invariant under arbitrary similarity transformations\footnote{That is, under any transformation of the form $A \longrightarrow P^{-1}AP$, where $P$ is a non singular matrix.}. In fact, $(-1)^k c_k$ is the sum of all the principal minors of order $k$ of $A$. In particular, $c_1=-\tr{A}$ and $c_n=(-1)^n \det{A}$.  
The profound interrelationship between a matrix  and its characteristic polynomial becomes even more surprising considering  the Cayley-Hamilton  theorem\footnote{The Cayley-Hamilton theorem states that every square matrix over a commutative ring (such as the real or complex numbers, or the integers) satisfies its own characteristic equation.} \cite{Horn} and/or the  Newton's identities \cite{kalman2000matrix}, which reveal the existence of finite algebraic  expressions for the coefficients of the characteristic polynomial of a matrix in terms of traces of powers (up to $n$) of the matrix. \cite{prasolov1994problems,boas2006mathematical}.

The function $C:\mathbb{M}_{n}\to\mathbb{P}_{n}[\mathbb{C}]$ is surjective
but not injective, and hence it cannot be inverted.  In fact, it is easy to convince oneself that, for any element $P_n(z)\in\mathbb{P}{_n}[\mathbb{C}]$, $C^{-1}({P}_n(z))$ is indeed an infinite subset of $\mathbb{M}_n$, since by definition it consists of all and only the matrices belonging to $M_n$ whose characteristic polynomial is ${P}_n(z)$.
\par Thus, while the direct or forward  problem of finding the characteristic polynomial of a given $n\times n$ matrix is certainly well-posed according to Hadamard \cite{hadamard1902problemes}, conversely,  the problem of finding a matrix $A\in \mathbb{M}_n$ generating a given complex polynomial $P_n(z)$ is an ill-posed inverse problem \cite{kabanikhin2008definitions,von2011ill}, as it manifestly violates Hadamard's uniqueness requirement, considering that every element $\in C^{-1}({P}_n(z))$ is a solution to the problem. 

It is possible to overcome such an ill-posedness by introducing a restriction $C\big|_{[C]_n}$ of the function $C$ to a subset $[C]_n$ of $\mathbb{M}{_n}$ which is injectively and surjectively valued on $\mathbb{P}{_n}[\mathbb{C}]$. To this end, let us first observe that the function $C$ is surjective and, by definition, ${P}_n(z)$ is the characteristic polynomial of all and only the matrices belonging to $C^{-1}({P}_n(z))\subset \mathbb{M}{_n}$. Moreover, $C^{-1}({P}_n(z))\cap C^{-1}({P'}_n(z))=\emptyset$ when ${P}_n(z)\neq {P'}_n(z)$. Therefore, the infinite subsets $C^{-1}({P}_n(z))$ of $\mathbb{M}{_n}$ corresponding to the infinite $n$-degree polynomials ${P}_n(z)$ represent a partition of $\mathbb{M}{_n}$. We can say equivalently that we are introducing in $\mathbb{M}{_n}$ the equivalence relation $A\sim B$ consisting in the condition that $A\in \mathbb{M}{_n}$ and $B\in \mathbb{M}{_n}$ share the characteristic polynomial and thus belong to a given equivalence class $C^{-1}({P}_n(z))$. At this point,we define the subset $[C]_n$ by choosing one element from each equivalence class. According to Zermelo's postulate,  $[C]_n\neq\emptyset$ can always be constructed (in infinitely many ways in the present case, since each equivalence class is infinite), and the cardinality of its intersection with $C^{-1}({P}_n(z))$ is precisely one for any ${P}_n(z)$ by construction.
Therefore, every \textit{one-to-one} function $C\big|_{[C]_n}: [C]_n\to\mathbb{P}{_n}[\mathbb{C}]$ obtained by applying the axiom of choice to the quotient set $\mathbb{M}{_n}/\sim$ to generate $[C]_n$ is invertible. Hence, the function $C\big|_{[C]_n}^ {-1}:\mathbb{P}{_n}[\mathbb{C}]\to[ C]_n$ defines a $[C]_n$-dependent Hadamard well-posed inverse problem, whose solution, by construction, can be given in terms of $[C]_n$ in the form
 \begin{eqnarray} C\big|_{[C]_n}^ {-1}{(P_n(z))}= [C]_n\cap C^{-1}({P}_n(z)). \label{solinv}\end{eqnarray}
 We remark that different legitimate choices of the subset $[C]_n$ lead to different  inverse problems, all well-posed in the fixed domain $\mathbb{P}{_n}[\mathbb{C}]$, and the corresponding solutions \eqref{solinv} differ in the generally non-similar images of one or more polynomials ${P}{_n}(z)$.
 \par  We also point out that any derogatory matrix $D$ cannot be classified as a companion matrix of its characteristic polynomial $P_D(z)\equiv \det(zI_n-D)$, since $D$ annihilates a polynomial having a degree lower than that of $P_D$ \cite{Horn}. 
\par In this paper, a matrix whose characteristic polynomial coincides with a given polynomial $P_n(z)$  is called an \textit{Almost-Companion  Matrix} of $P_n(z)$. Clearly, any CM of $P_n(z)$ is an ACM too. A derogatory matrix $D$ such that $P_D(z)=P_n(z)$ is an ACM. In addition, a matrix similar to an ACM is still an ACM. The converse of this statement is generally false: two ACMs of the same given polynomial are not necessarily similar  \cite{Barnett}.  The set of all the ACMs of $P_n(z)$ cannot be generated by similarity transformations starting from an assigned ACM, since this set always includes both derogatory and nonderogatory matrices.\label{inverse problem}

\section{Almost-Companion Matrices of a Cubic Complex Polynomial}\label{sec:cubic}

In this section, we focus on the search of an ACM for the third degree polynomial 
\begin{eqnarray}P_{3c}(\eta)=\eta^3 +p\eta +q, \label{cubic canonical}\end{eqnarray}
which is the canonical form of 
\begin{eqnarray}P_3(z)&=&z^3 +c_1z^2+c_2z +c_3, \label{cubic}  \end{eqnarray}
obtained by the translation  
\begin{eqnarray}\eta=z+\frac{c_1}{3}.\label{translation}\end{eqnarray} 
The generally complex numbers $p$ and $q$ in \eqref{cubic canonical} are related to the coefficients of $P_3(z)$ as follows:
\begin{eqnarray} p= -\frac{c_1^2}{3} + c_2,\quad q=\frac{ 2c_1^3}{27} -\frac {c_1c_2}{3}+c_3.\label{pq}\end{eqnarray}
 We denote $Q_{3c}$ the ACM of $P_{3c}(\eta)$ defined by
\begin{eqnarray}P_{3c}(\eta)&\equiv& \det(\eta I_3-Q_{3c})=\det((\eta  -\textstyle\frac{c_1}{3})I_3-(Q_{3c}-\frac{c_1}{3}I_3))\nonumber\\&=& \det(z I_3-(Q_{3c}-\textstyle\frac{c_1}{3}I_3))\equiv P_3(z),\label{ QQ}\end{eqnarray} 
which means that
\begin{eqnarray}Q_{3}=Q_{3c}-\textstyle\frac{c_1}{3}I_3,\label{Q3Q3c}\end{eqnarray}
is the simple recipe to get the corresponding ACM $Q_{3}$ of $P_3(z)$ from $Q_{3c}$. This analysis sheds light on the advantage of first deriving $Q_{3c}$ for the simpler canonical form of a given polynomial and then finding $Q_{3c}$ from the straightforward relation (9).
Next, we formulate a trial ACM $Q_{3c}$ of \eqref{cubic canonical}. To this end, we observe that, in accordance with the Vieta-Girard formula for the sum of the roots of \eqref{cubic canonical} \cite{10.2307/2299273}, the absence of the quadratic term in $P_{3c}(\eta)$ implies that $\mbox{tr}(Q_{3c})=0$. 
Moreover, every matrix with elements in $\mathbb{C}$ is unitarily equivalent to a matrix with equal main diagonal elements \cite{Horn}. Thus, it is legitimate to set the diagonal elements of our trial $Q_{3c}$ equal to zero. In constructing an ACM of \eqref{cubic canonical}, we aim to write its  non diagonal elements in such a way that, in the particular case of a real cubic $P_3$, and hence $P_{3c}$, the trial matrix $Q_{3c}$ becomes structurally Hermitian provided that $p$ and $q$ in \eqref{cubic canonical} satisfy specific conditions, which will also be derived within our approach. The feasibility of this approach will highlight the greater flexibility of the ACMs compared to that of the CMs.

Following this strategy, we propose the following trial $Q_{3c}$:
\begin{eqnarray}Q_{3c}\equiv Q_{3c} (\rho,\varphi,\varphi_{13}) =-\left(
                     \begin{array}{ccc}
                       0 & \rho e^{i\frac{\varphi}{2}} &  \rho e^{i\frac{\varphi}{2}}                     e^{i\varphi_{13}} \\
                        \rho e^{i\frac{\varphi} {2}}  & 0 &  \rho  e^{i\frac{\varphi}{2}}  \\
                         \rho  e^{i\frac{\varphi}{2}}                     e^{-i\varphi_{13}}  &\rho  e^{i\frac{\varphi}{2}}       & 0 \\
                     \end{array}
                   \right),\label{Q3C}
\end{eqnarray}
where the minus sign was introduced for convenience, considering the form of the charateristic polynominal. In equation \eqref{Q3C}, $\rho$ is real and positive, $\varphi$ is real, whereas $\varphi_{13}$ is, in general, a complex number. It is readily seen that, when $\varphi=0$ or $\pi$ and $\varphi_{13}$ is real, 
the matrix $Q_{3c}(\rho,\varphi,\varphi_{13})$ is Hermitian, consistent with our search strategy. It is useful to note that the complex conjugate of $e^{i\varphi_{13}}$ is $e^{-i\varphi_{13}^{\star}}$, where $\varphi_{13}^{\star}$ denotes the conjugate of $\varphi_{13}$.

The characteristic polynomial of $Q_{3c} (\rho,\varphi,\varphi_{13})$ is
\begin{eqnarray}\det{(\eta I_3- Q_{3c}(\rho,\varphi,\varphi_{13}))}=\eta^3-3\rho^2 e^{i\varphi}\eta+2\rho^3 e^{\frac{3}{2}i\varphi} \cos{\varphi_{13}}.\label{polyQ3C}
\end{eqnarray}
Then, identifying the polynomial \eqref{cubic canonical} with \eqref{polyQ3C} yields:
\begin{eqnarray}\label{pqconditions}
 \left\{
\begin{array}{l}
 p\equiv |p|e^{i\Theta_p}= -3\rho^2 e^{i\varphi}=3\rho^2 e^{i(\varphi+\pi)},
\\
q=2\rho^3 e^{{\frac{3}{2}}i\varphi} \cos{\varphi_{13}}.
\end{array}
\right.
\end{eqnarray}
Given $p$, the first equation \eqref{pqconditions} allows us to fix $\rho$ and select $\varphi$ (in an infinite set) as follows:
\begin{eqnarray} \rho=\sqrt{\frac{|p|}{3}} ,\quad  \varphi=\Theta_p-\pi.\label{pefi}\end{eqnarray}
Defining, for $p\neq0$, the complex parameter
\begin{eqnarray}\chi=\frac{-i q e^{-\frac{3}{2}i\Theta_p}}{2\sqrt{\frac{|p|^3}{27}}},  \label{cosfi13}\end{eqnarray}
the second equation \eqref{pqconditions} becomes an elementary trigonometric equation in $\mathbb{C}$: 
\begin{eqnarray}
\cos{\varphi_{13}}=\chi,
\label{equperphi13}\end{eqnarray}
which admits infinitely many solutions for any $\chi$; in fact, similarly to the cosine function of a real variable, the complex cosine function is even and periodic with period $2\pi$. 

Using Euler's formula, equation \eqref{equperphi13} is easily transformed into a quadratic equation in the variable  $e^{i\varphi_{13}}$, whose solution leads to
\begin{eqnarray}
	\varphi_{13}=  -i\ln\left(\chi+ {i|1-\chi^2|}^\frac{1}{2}e^{\frac{i}{2}\arg(1-\chi^2)}\right)=\arccos\chi.\label{newsolphi13}\end{eqnarray}
Due to the presence of the multi-valued complex function $\arg(\chi^2-1)$, expression \eqref{newsolphi13} represents the set of infinite images of $\chi$ generated by the inverse of the non-injective cosine function
over $\mathbb{C}$.
Therefore, strictly speaking, the  expression found for $\varphi_{13}$ cannot be introduced as it is in the matrix $Q_{3c} (\rho,\varphi,\varphi_{13})$. In fact, the three parameters appearing in the trial ACM \eqref{polyQ3C} of \eqref{cubic canonical} must be single-valued functions of the complex coefficients $p$ and $q$. 
For our purposes, therefore, we now need to extract a specific single-valued complex function from the multi-valued function $\varphi_{13}$. 
 
The single-valued complex function that we use here is the principal value $\Phi_{13} $ of $\varphi_{13}$, which is obtained from \eqref{newsolphi13} by substituting the multi-valued functions $\arg$ and $\ln$ with their  principal values, denoted $\operatorname{Arg}$ and $\operatorname{Ln}$, respectively. This choice amounts, by definition, to constructing the principal value $\operatorname{Arccos}{(\chi)}$ of function $\arccos(\chi)$, which is mostly used in the literature \cite{abramowitz1964handbook,haber2011complex}. It is worth noting that equation \eqref{newsolphi13} can also be written in terms of $\chi^2-1$, but then the use of the principal value in the resulting expression for $\varphi_{13}$ would have some drawbacks, as is discussed in detail in \cite{kahan1987branch}.
The procedure described above gives
\begin{eqnarray}
 \Phi_{13}&\equiv&\Phi_{13}(\chi)= -i\operatorname{Ln}\left(\chi+ i|1-\chi^2|^{\frac{1}{2}}e^{\frac{i}{2}\operatorname{Arg}           (1-\chi^2)}\right)\equiv\operatorname{Arccos}{(\chi)}\nonumber\\&=& \operatorname{Arg}\left(\chi+ i|1-\chi^2|^{\frac{1}{2}}e^{\frac{i}{2}\operatorname{Arg}(1-\chi^2)}\right)-i\operatorname{\textit{ln}}\left(\left|\chi+ i|1-\chi^2|^{\frac{1}{2}}e^{\frac{i}{2}\operatorname{Arg}(1-\chi^2)}\right|\right)\nonumber\\&\equiv&  Re(\Phi_{13})+iIm(\Phi_{13}), 
 \label{FI13}\end{eqnarray}
 where $\operatorname{\textit{ln}}$ denotes the ordinary real logarithm of its positive argument and the first term in \eqref{FI13} is the principal value of  $\operatorname{arg}(\chi+ {i|1-\chi^2|}^{\frac{1}{2}}e^{\frac{i}{2}\operatorname{Arg}(1-\chi^2)})$, which, by definition, generates real images in $ (-\pi,\pi]$. Equation \eqref{FI13} provides the algebraic representation of the complex single-valued function $\operatorname{Arccos}{(\chi)}$ whatever $\chi\in \mathbb{C}$. Note that, for any real $\chi$ such that  $|\chi|\leq 1$ $(|\chi|\geq 1)$, the immaginary (real) component of $\Phi_{13}$ identically vanishes.
 
We have determined all the ingredients for constructing the trial ACM of \eqref{cubic canonical} when $p\neq 0$. In accordance with \eqref{FI13}, this ACM is a generally non-Hermitian matrix that can be written as follows:
 \begin{eqnarray}\widetilde{Q}_{3c} (p,q) =\sqrt{\frac{|p|}{3}}e^{i\frac{\varphi_p}{2}}\left(
                     \begin{array}{ccc}
                       0 &  1 &                        e^{i\Phi_{13}(\chi)} \\
                         1 & 0 &  1  \\
                                               e^{-i\Phi_{13}(\chi)}  &   1       & 0 \\
                     \end{array}
                   \right).\label{Qpq}
\end{eqnarray}
where
\begin{eqnarray}\varphi_p=\Theta_p+\pi.\label{pefi}\end{eqnarray}
Thus, based on \eqref{Q3Q3c}, the ACM of \eqref{cubic} has the form 
 \begin{eqnarray}\widetilde{Q}_{3} (p,q) =\left(
                     \begin{array}{ccc}
                        -\frac{c_1}{3} &  \sqrt{\frac{|p|}{3}}e^{i\frac{\varphi_p}{2}} &                        \sqrt{\frac{|p|}{3}}e^{i\left[\frac{\varphi_p}{2}+\Phi_{13}(\chi)\right]} \\
                         \sqrt{\frac{|p|}{3}}e^{i\frac{\varphi_p}{2}} & -\frac{c_1}{3}   &  \sqrt{\frac{|p|}{3}}e^{i\frac{\varphi_p}{2}}  \\
                                               \sqrt{\frac{|p|}{3}}e^{i\left[\frac{\varphi_p}{2}-\Phi_{13}(\chi)\right]}  &   \sqrt{\frac{|p|}{3}}e^{i\frac{\varphi_p}{2}}       & -\frac{c_1}{3}   \\
                     \end{array}
                   \right).\label{Q3_1}
\end{eqnarray}
We can find an ACM of \eqref{cubic canonical} for $p=0$ through the same kind of approach, beginning with a matrix different from \eqref{Q3C}. Our solution, denoted by $\overline{Q}_{3c} (q)$, can be cast as follows:
 \begin{eqnarray} \overline{Q}_{3c} (q)=\left(\frac{|q|}{\sqrt{3}}\right)^\frac{1}{3}e^{\frac{i}{3}\operatorname{Arg}(iq)}\left(                    \begin{array}{ccc}
                       0 &  1 &                       1 \\
                 -e^{-i\frac{4}{3}\pi}         & 0 &  -1  \\
                                                 -e^{i\frac{4}{3}\pi} &   1       & 0 \\
                     \end{array}
             \right).\label{Qq}
\end{eqnarray}
For completeness, we also write $\widetilde{Q}_{3c}(p,0)$ $(\chi=0)$:
\begin{eqnarray}\widetilde{Q}_{3c} (p,0) =\sqrt{\frac{|p|}{3}}e^{i\frac{\varphi_p}{2}}\left(
                     \begin{array}{ccc}
                       0 &  1 &                         i \\
                         1 & 0 &  1  \\
                                                -i  &   1       & 0 \\
                     \end{array}
                   \right),\label{Qp0}
\end{eqnarray}
since $\Phi_{13}(0)=\pi/2$ from \eqref{FI13}. It is not difficult to see that the eigenvalues of \eqref{Qp0} are $0$ and  $\pm|p|^{\frac{1}{2}}e^{\frac{i\varphi_p}{2}}$.

Above, we formulated and solved the inverse problem of finding ACMs of generic cubic complex polynomials. Next, exploiting these ACMs, we will present a way to find the roots of the given polynomial without resorting to the Cardano-Dal Ferro formulas, and then we will delve into the implications of having a polynomial with real coefficients on the form of the ACM and its roots.
\subsection{Roots Characterization}\label{discussion}
\par  Next we assume that $p \neq 0$. Then, the eigenvalues of $\widetilde{Q}_{3c}(p,q)$ are the roots of \eqref{cubic canonical}, and hence the roots of the characteristic polynomial of the matrix appearing in the right-hand side of \eqref{Qpq}, each multiplied by the pre-factor $\sqrt{\frac{|p|}{3}}e^{i\frac{\varphi_p}{2}}$. 
This characteristic polynomial $\tilde{P}_{3c}(\tilde\eta)$ in the unknown $\tilde\eta$  has the form
\begin{eqnarray}\tilde{P}_{3c}(\tilde\eta)= {\tilde\eta}^3 -3{\tilde\eta}-2\cos(\Phi_{13}(\chi)). \label{cubictilde} \end{eqnarray}

The cosine representation of the free term in the polynomial \eqref{cubictilde} is remarkable, because it allows one to guess, at first glance, one of its three roots and then to exactly construct the other two by simply reducing the cubic polynomial $\tilde{P}_{3c}(\tilde\eta)$ to a quadratic polynomial. In fact, without resorting to the well-known Cardano-Dal Ferro formulas \cite{doi:10.1142/1284}, and using instead the elementary triplication formula for the cosine function $\cos(3z)= 4\cos^3(z)-3\cos(z)$, which also holds in the complex field, it is immediate to see that the generally complex expression\footnote{An equivalent solution was obtained in \cite{Lambert}, using an angle offset by $-\pi$ with respect to $\Phi_{13}$. Solutions in trigonometric form of the canonical (or depressed) cubic equation equivalent to those presented here and in \cite{Lambert} were obtained by Francois Viete (1540-1603) for the case in which the polynomial in equation \eqref{cubic canonical} is real.}
\begin{eqnarray}\tilde\eta_{1}= 2\cos\left(\textstyle{\frac{1}{3}}\Phi_{13}(\chi)\right) \label{magicroot}
\end{eqnarray} 
is a root of \eqref{cubictilde} whatever the complex coefficients $p \neq 0$ and $q$. The algebraic representation \eqref{FI13} of $\Phi_{13}(\chi)$ is the key to explicitly write the $p$- and $q$-dependence of the real and imaginary components of the algebraic expression for $\tilde\eta_1$, which are obtained using Euler's formula as
\begin{eqnarray}  
Re(\tilde\eta_1)=2\cos\left({\frac{1}{3}}Re(\Phi_{13})\right) \cosh\left({\frac{1}{3}}Im(\Phi_{13})\right) ,
\\
Im(\tilde\eta_1)=-2\sin\left({\frac{1}{3}}Re( \Phi_{13})\right)  \sinh\left({\frac{1}{3}}Im(\Phi_{13})\right),
\label{ReIm}
\end{eqnarray}
where  $Re( \Phi_{13})$ and $Im(\Phi_{13})$ are defined in \eqref{FI13}.
The other two roots are easily found to be \cite{Lambert}
\begin{eqnarray}\tilde\eta_{k}= -\frac{1}{2}\tilde\eta_{1}+(-1)^{k+1} \sqrt{3}|{\sin}^2\left(\textstyle{\frac{1}{3}}\Phi_{13}(\chi)\right)|^\frac{1}{2}e^{\frac{i}{2}\operatorname{Arg}\sin^2\left(\textstyle{\frac{1}{3}}\Phi_{13}(\chi)\right)},\quad k=2,3.
\label{etapmroot}
\end{eqnarray}
Equations \eqref{magicroot} and \eqref{etapmroot} express the roots of $\tilde{P}_{3c}(\tilde\eta)$ as functions of parameter $\Phi_{13}$, which appears in the first and last antidiagonal terms of matrices \eqref{Qpq} and \eqref{Q3_1}. We thus conclude that our procedure to construct the ACM also yields the roots of the (generally complex) characteristic polynomial. 
It is interesting to highlight the conditions for a complex polynomial to admit real roots (note that the roots can never be all real, however, if the imaginary part of at least one of the polynomial coefficients is nonzero). To this end, it is convenient to start from the polynomial form \eqref{cubic}. We write the three coefficients of \eqref{cubic} as $c_j\equiv x_j+iy_j$, with $j=1,2,3$. If a real root $r$ of \eqref{cubic} exists, it must satisfy the equation $y_1r^2+y_2r+y_3=0$, which results from equating to zero the imaginary part of the polynomial. $\delta=(y_2)^2-4y_1y_3\geq 0$ is clearly a necessary condition for the existence of the root $r$, and the two only possible expressions of $r$ are $\frac{-y_2\pm\sqrt{\delta}}{2y_1}$. Then, any of these two expressions is indeed a root of \eqref{cubic} only if it satisfies the additional condition $r^3+x_1r^2+x_2r=-x_3$, which results from the real part of the polynomial. We can thus state that the cubic polynomial \eqref{cubic} has at least one real root if and only if the inequality $\delta\geq 0 $ and the last condition are both met. In particular, when $\delta= 0 $, $r$ is a double root. Finally, we examine the case $y_1=0$ (and $y_2\neq 0$, since otherwise no real root exists for a complex polynomial). Applying the same procedure, one finds that a real root $r=-\frac{y_3}{y_2}$ of \eqref{cubic} exists if and only if the condition $r^3+x_1r^2+x_2r=-x_3$ holds.
It is easy to convince oneself that this real root has multiplicity two, being a real root of the first derivative of \eqref{cubic}.
 
\subsection{Real polynomial case}
\par We now investigate the special form of $\widetilde{Q}_{3c}(p,q)$ in the case in which the three coefficients of \eqref{cubic} are real, with the aim of establishing properties of the polynomial roots based on its almost-companion representation built above.  

In this case, \eqref{pq} implies that $p$ and $q$ are real, and thus \eqref{cubic canonical} is also a real polynomial over $\mathbb{C}$. Moreover, since $\Theta_p$ can only be $0$ ($p>0$) or $\pi$ ($p<0$), the parameter $\chi$ defined in \eqref{cosfi13} is purely imaginary or real, respectively, and therefore its square
\begin{eqnarray}
\chi^2=-\frac{27q^2 e^{-3i\Theta_p}}{4|p|^3}=-\frac{27q^2}{4p^3},
\label{FI13real}\end{eqnarray}
 is real for any $p$. As a consequence, using \eqref{FI13} it is not difficult to prove that, if and only if the discriminant\footnote{Note that the discriminant can also be defined with the opposite sign \cite{Dickson}. Clearly in this case all the inequalities involving the discriminant of the polynomial in canonical form are to be inverted.}
 \begin{eqnarray}\Delta(p,q)\equiv\frac{p^3}{27}+\frac{q^2}{4}\leq 0,
\label{deltanonpositivo} 
\end{eqnarray}
(which implies $p<0$, that is, $\Theta_p=\pi$, and $|\chi|\leq 1$), the imaginary part of $\Phi_{13}(\chi)$ given in equation \eqref{FI13} vanishes, while its real part assumes the simple expression


\begin{eqnarray}
 \Phi_{13}(p,q)=    \operatorname{Arg}\left(\chi+ i\Big|1-\chi^2\Big|^{\frac{1}{2}} \right)=\pi -\operatorname{Arccos}\left(- \frac{3q}{2p}\sqrt\frac{-3}{p}\right),
\label{FI13real_2}
\end{eqnarray}
where we used the identity $\operatorname{Arccos}(x)=\pi-\operatorname{Arccos}(-x)$ valid for any real $x$ such that $|x|\leq 1$. It is worth noting that in the present case \eqref{etapmroot} takes the simpler form $\tilde\eta_{k}= -\frac{1}{2}\tilde\eta_{1}+(-1)^{k+1} \sqrt{3} \sin\left(\textstyle{\frac{1}{3}}\Phi_{13}(\chi)\right) $  since \eqref{FI13real_2} shows that $\Phi_{13}\in [0,\pi]$.

Under the condition \eqref{deltanonpositivo} and considering that $\varphi_p=\Theta_p+\pi=2\pi$, the specialization of \eqref{Qpq} to the case under scrutiny produces the Hermitian ACM of \eqref{cubic canonical} as
\begin{eqnarray}\widetilde{Q}_{3c}(p,q) =-\sqrt{\frac{|p|}{3}}\left(
                     \begin{array}{ccc}
                       0 &  1 &                        e^{i\Phi_{13}(p,q)} \\
                         1 & 0 &  1  \\
                                               e^{-i\Phi_{13}(p,q)}  &   1       & 0 \\
                     \end{array}
                   \right),\label{Qpq_1}
\end{eqnarray}
where the real angle $\Phi_{13}$ is given by equation \eqref{FI13real_2}. The Hermitian nature of the ACM built assures that \eqref{cubic canonical}, as well as \eqref{cubic}, has three real roots. These roots are distinct when $\Delta(p,q)<0$, while two of them are coincident if $\Delta(p,q)=0$ \cite{Tricomi}.

We can write the real roots $x_k$ $(k=1,2,3)$ of \eqref{cubic} as follows:
\begin{eqnarray}x_k=-\sqrt{\frac{|p|}{3}}\tilde\eta_k - \frac{c_1}{3},  \label{roots z}\end{eqnarray}
where  $ \tilde\eta_k$ are the three (real) roots of \eqref{cubictilde}.
Equations \eqref{magicroot} and \eqref{etapmroot}, together with \eqref{FI13real_2}, yield
\begin{eqnarray}x_k=2\sqrt{\frac{|p|}{3}}\cos\left(\frac{\Phi_{13}(p,q)+(2k+1)\pi}{3}\right)-\frac{c_1}{3}, \quad k=1,2,3. \label{radicireali} \end{eqnarray}
 It is possible to check that this formula gives the well-known trigonometrical and translated forms of the three roots of \eqref{cubic} when they are real (see \cite{doi:10.1080/00029890.1992.11995845}). 

When $\Delta(p,q)>0$, only one of the three roots of Equation \eqref{cubic canonical} or \eqref{cubic} (with real polynomial coefficients) is real, while the other two roots are complex conjugate. In particular, the real root corresponds to $\tilde\eta_3$ for $p>0$ and to $\tilde\eta_1$ for $p<0$. In more detail, the three roots are
\begin{align}
	z_1 &= \frac{\sqrt{p}}{2}(Y+iX)-\frac{c_1}{3} \label{rootspg01}
	\\
	z_2 &= \frac{\sqrt{p}}{2}(Y-iX)-\frac{c_1}{3} \label{rootspg02}
	\\
	z_3 &= -\sqrt{p}Y \label{rootspg03}-\frac{c_1}{3}
\end{align}
where
\begin{equation}
	X = \sqrt[3]{\sqrt{1-\chi^2}+i\chi}+\sqrt[3]{\sqrt{1-\chi^2}-i\chi}
	\label{rootspg04}
\end{equation}
and
\begin{equation}
	Y = \frac{\sqrt[3]{\sqrt{1-\chi^2}+i\chi}
	-\sqrt[3]{\sqrt{1-\chi^2}-i\chi}}{\sqrt{3}}
	\label{rootspg05}
\end{equation}
with
\begin{equation}
\chi=-i\frac{3q}{2p}\sqrt{\frac{3}{p}}
\label{rootspg06} \end{equation}
for $p>0$, and
\begin{align}
	z_1 &= -\sqrt{-\frac{p}{3}}C-\frac{c_1}{3} \label{rootspl01}
	\\
	z_k &= \sqrt{-\frac{p}{3}}\left[\frac{C}{2}
	+i(-1)^k\sqrt{3\left(\frac{C^2}{4}-1\right)}\right]-\frac{c_1}{3}, \quad k=2,3,
	\label{rootspl02}
\end{align}
where
\begin{equation}
	C=\sqrt[3]{\chi+\sqrt{\chi^2-1}}+\sqrt[3]{\chi-\sqrt{\chi^2-1}}
	\label{rootspl03}
\end{equation}
with
\begin{equation}
	\chi=-\frac{3q}{2p}\sqrt{-\frac{3}{p}}
	\label{rootspl04} \end{equation}
for $p<0$ (see derivation in Appendix A).\footnote{A more elaborate derivation of the roots leading to their formulation in terms of trigonometric functions can be found in \cite{Tricomi}.}

\section{Almost-Companion Density Matrices of a Qutrit on Demand}\label{densitymatrix}
  
A quantum system living in the  Hilbert space $\mathcal{H}$ spanned by three orthonormal states $\ket1$, $\ket2$, and $\ket3$ is called a qutrit \cite{carmichael2001directions}.
A pure state of the qutrit can always be represented as a normalized linear combination of  these three states.
To describe an arbitrary pure or mixed state of the qutrit with the same formalism, one uses instead a linear operator $\hat \rho$ called the density operator \cite{fano1957description}. 
 It acts on $\mathcal{H}$ and, by definition, is positive semi-definite with trace 1: $\operatorname{\textit{tr}}(\rho)=1$. It is well known that any positive semi-definite operator is Hermitian since its skew-Hermitian part vanishes \cite{Horn,blanchard2015mathematical}. As a consequence, any positive semi-definite operator is diagonalizable, and it is possible to show that its eigenvalues are real non-negative numbers. In particular, any density operator $\hat \rho$ is Hermitian. The three eigenvalues of the operator $\hat \rho$ describing the state of a qutrit are the populations of the three eigenstates of $\hat \rho$.
 The $3\times 3$ basis-dependent matrix representation of $\hat \rho$, called density matrix and denoted by $\rho$, is also positive semi-definite and hence Hermitian. We observe incidentally that, conversely, any Hermitian matrix with non-negative eigenvalues is positive semi-definite and, if its trace is 1, it is a density matrix.

The purpose of this section is to demonstrate that our recipe for constructing the Hermitian ACM of a generic third degree polynomial admitting three real roots provides an effective tool for writing density matrices of a qutrit on demand.
 
The aforementioned definition of density matrix results in unambiguous properties of the real coefficients of its characteristic polynomial. First of all, writing this polynomial in the canonical form
\begin{eqnarray}p_{3c}(\eta)=\eta^3 +p\eta +q, \label{pcubic canonical}\end{eqnarray} 
the condition
\begin{eqnarray}p\leq -\frac{3}{2}\sqrt[3]{2q^2} \label{1629}\end{eqnarray}
stemming from Equation \eqref{deltanonpositivo} ensures that $\Phi_{13}(p,q)$ is real, so that the ACM \eqref{Qpq_1} of \eqref{pcubic canonical} is Hermitian and hence has real eigenvalues. Then, turning to form \eqref{cubic} of the monic polynomial through the translation \eqref{translation}, the Vieta-Girard formula for the sum of the roots \cite{connor1956historical} implies that the coefficient of the quadratic term be $-1$. Moreover, Descartes's sign rule \cite{descartes2012geometry} requires that the four coefficients of polynomial \eqref{cubic} have alternate signs in order to have three positive routes. Therefore, the characteristic polynomial of an arbitrary density matrix of a qutrit is necessarily a third-degree real and monic polynomial of the form
\begin{eqnarray}p_3(x)=x^3 -x^2+a^2x-b^2,
\label{polydensity}\end{eqnarray}
where $a$ and $b$ are real numbers that satisfy condition \eqref{1629} after translation \eqref{translation}. One or two roots are zero if $a\neq0, b=0$ or $a=0, b=0$, respectively, while the inequality \eqref{1629} is never satisfied for $a=0$ and $b\neq0$.

In conclusion, under the conditions \eqref{1629} and \eqref{polydensity},
$\widetilde{Q}_{3}(p,q)=\widetilde{Q}_{3c}(p,q)-\textstyle\frac{c_1}{3}I_3$, with $\widetilde{Q}_{3c}(p,q)$ given by Equation \eqref{Qpq_1}, is a density matrix. Incidentally, we point out that our inverse problem admits infinitely many non-Hermitian solutions, that is, non-Hermitian ACMs of \eqref{pcubic canonical} or \eqref{polydensity}, such as, for example, the corresponding Frobenius companion matrix. Therefore, the explicit construction of the Hermitian ACM of \eqref{polydensity} and, more generally, of any real third-degree polynomial with only real roots is a successful outcome of our search strategy \eqref{Q3C}. This recipe, in turn, forms the basis of the application presented below.

Let us introduce the set $\mathcal{D}$ of all density matrices of a qutrit in a given basis $\{\ket n, n= 1,2,3\}$ of $\mathcal{H}$. $\mathcal {E}$ be the binary relation  in $\mathcal{D}$ defined as follows: $\rho_1 \in \mathcal{D}$ and $\rho_2\in\mathcal{D} $ are in the relation $\mathcal {E}$ if they are ACMs of the same polynomial $p_3(x)$ defined in \eqref{polydensity}. This relation, expressed by writing $\rho_1\mathcal {E}\rho_2$, is an equivalence relation as it is manifestly reflexive, symmetric, and transitive. 
$\mathcal{D}$ is thus partitioned by $\mathcal {E}$. The quotient set $\mathcal{D}/\mathcal {E}$  consists of all equivalence classes of $\mathcal{D}$ with respect to $\mathcal {E}$. Each equivalence class, which comprises all density matrices with the same characteristic polynomial, is represented by one (arbitrarily chosen) of its elements, $\overline\rho$, and is commonly denoted by $[\overline\rho]$. This is where our result \eqref{Qpq_1} enters the scene, providing an easy way to parameterize the quotient set of $\mathcal{D}$. 

It is always possible to use the matrix $\widetilde{Q}_{3}(p,q)=\widetilde{Q}_{3c}(p,q)-\textstyle\frac{c_1}{3}I_3$ as the representative element of the  equivalence class consisting of all elements of $\mathcal{D}$ sharing the characteristic polynomial \eqref{polydensity}, which, in turn, is uniquely associated to its canonical form \eqref{pcubic canonical}. In this way, we establish a one-to-one correspondence between $\mathcal{D}/\mathcal {E}$ and the set $\mathcal{P}$ of polynomials \eqref{pcubic canonical}. This correspondence amounts to parameterize the quotient set of $\mathcal{D}$ in terms of $p$ and $q$. The most ambitious target of parameterizing the set $\mathcal{D}$ is discussed in a recent topical issue \cite{bruning2012parametrizations}.
It is worth emphasizing that a density matrix $\rho$ belongs to the class of equivalence $\left[\widetilde{Q}_{3c}(p,q)-\textstyle\frac{c_1}{3}I_3\right]$ if and only if it can be unitarily generated from $\widetilde{Q}_{3c}(p,q)-\textstyle\frac{c_1}{3}I_3$, since its characteristic polynomial, trace, and Hermiticity are unitarily invariant, thus implying the invariance of the positive semi-definiteness.
Therefore, while two similar matrices are ACMs of the same polynomial, a similarity transformation  of a density matrix does not generate, in general, a density matrix \cite{Horn}. 

Our parameterization of $\mathcal{D}/\mathcal {E}$ in terms of the coefficients of its characteristic polynomial written in canonical form provides the theoretical basis for constructing, on demand and in a prefixed basis of $\mathcal{H}$, almost-companion density matrices of any assigned polynomial $p_3(x)$ fulfilling the condition $\Delta \leq 0$.
We illustrate the concrete applicability of our recipe by constructing an almost-companion density matrix starting from the polynomial
\begin{eqnarray}p_3(x)=x^3 -x^2+\frac{11}{36}x-\frac{1}{36}. 
\label{example}\end{eqnarray} 
 The translation: $\eta=x-\frac{1}{3}$ yields 
\begin{eqnarray}p_{3c}(x)=\eta^3 -\frac{1}{36}\eta , \label{examplec}\end{eqnarray} 
so that, in this case, $p=-\frac{1}{36}$ while $q$ vanishes.
Exploiting Equation \eqref{FI13real_2}, we easily get:
\begin{eqnarray}
 \Phi_{13}(-\frac{1}{36},0)=     \pi -\operatorname{Arccos}(0)= \frac{\pi}{2}.
\label{exFI13real_2}
\end{eqnarray}
We have thus obtained the few ingredients necessary to build an almost-companion density matrix of the given polynomial \eqref{example} as the sum $\widetilde{Q}_3$ of the representative element of the corresponding equivalence class and the matrix $\frac{1}{3}I_3$, in accordance with the realization \eqref{Q3_1} of \eqref{Q3Q3c}.
That is, the density matrix has the form
\begin{eqnarray}\widetilde{Q}_{3} =\frac{1}{6\sqrt{3}}\left(
                     \begin{array}{ccc}
                        2\sqrt{3} & -1  & -i                        \\
                         -1  & 2\sqrt{3} & -1   \\
                                             i    &  -1       & 2\sqrt{3} \\
                     \end{array}
                   \right).\label{Q3}
\end{eqnarray}
Note that, while in this case it is easy to find the roots of \eqref{examplec} directly, and then those of \eqref{example}, in general the roots of the polynomial $p_{3c}(\eta)$ corresponding to a given $p_{3}(x)$ can be found using \eqref{radicireali}.

We stress that any matrix equivalent to \eqref{Q3} through a unitary transformation $\hat V$ is an almost-companion density matrix of \eqref{example} and vice versa. For a given basis, each density matrix thus obtained describes a different (generally mixed) state of the qutrit. If, instead, the unitary transformation is interpreted as the generator of a change of the basis $\{|n>, n=1,2,3\}$, the matrix obtained represents the same density operator in the new basis $(\hat V|n>, n=1,2,3)$.

\section{Unitary Matrices (Operators) on Demand}\label{unitarymatrix}
\par The effective construction of density matrices on demand in Section \ref{densitymatrix}, results from the application of our procedure for constructing ACMs to third degree polynomials that belong to the set $\mathcal{P}$ and satisfy, a priori, necessary and sufficient conditions for the existence of positive semi-definite ACMs of trace 1.
By comparison, the construction of almost-companion unitary matrices on demand (that is, starting from a given appropriate third degree polynomial) requires addressing two hurdles. The first one is to establish with certainty whether the given polynomial can be the characteristic polynomial of a unitary matrix without knowing its zeros a priori. The second difficulty lies in the fact that the trial ACM of a complex arbitrary polynomial, as given by the main equations \eqref{Q3Q3c} and \eqref{Q3C} of our procedure, is never unitary by construction. In this regard, it is important to note that the possibility of finding a non-unitary ACM of a given polynomial is not incompatible with the existence of a unitary almost-companion matrix for the given polynomial. In fact, different ACMs of a given polynomial are generally not unitarily equivalent.

In the light of these considerations, we want first to identify possible structural properties shared by the coefficients of all the characteristic polynomials of a unitary matrix. Then, according to our general procedure, we will introduce a class of trial unitary matrices sufficiently representative to allow us to find a unique ACM for an assigned polynomial whose three roots are unknown but certainly have modulus 1. 

\subsection{Properties of the Characteristic Polynomial of a Unitary Matrix} \label{U41}
\par It is easy to prove the following necessary and sufficient condition concerning the characteristic polynomial of a unitary ACM:
\begin{thm} \label{Pn}
Let $D_m(z)$ be any complex polynomial of degree $m$, with $1\leq m\leq n$, dividing an arbitrarily given complex polynomial $P_n(z)$. Then $P_n(z)$ admits a unitary ACM if and only if any $D_m(z)$ does.
\end{thm}
\begin{proof}
Necessity: if $P_n(z)$ admits a unitary ACM, then all its roots have modulus one. This property is obviously transferred to each $D_m(z)$  dividing $P_n(z)$ which, in turn, implies the existence of a diagonal unitary ACM of $D_m(z)$.

Sufficiency:  Since $P_n(z)$ can be represented as product of $n$ monic binomials whose free terms are complex numbers of modulus one by hypothesis, then a diagonal ACM of $P_n(z)$ with entries having modulus one exists. This ACM is unitary \cite{Horn}.
\end{proof}

When $n=2$, it is easy to convince oneself that
\begin{thm} \label{tp2}
A monic complex, second-degree polynomial is the characteristic polynomial of a unitary matrix of order 2, if and only if it has the structure  \begin{eqnarray}
\mathcal{P}_{2}(z) = z^2-r_2e^{i\vartheta}z+e^{2i\vartheta},
\label{P2}\end{eqnarray}
with $r_2\in [0,2]$ and $\vartheta\in (-\pi,\pi]$.
\end{thm}
\begin{proof}
To demonstrate this double statement it is sufficient to explicitly find the two roots of \eqref{P2} for the necessity and to use simple geometric arguments (or exploit Theorem \ref{Pn}) for the sufficiency. 
\end{proof}

We additionally remark that, for $r_2>2$ , the principal arguments of the two roots of \eqref{P2} coincide with $\chi$, and the product of their modules, both different from unity, is still one.

When the order of the unitary matrix is greater than $2$, it is still possible to find peculiar properties possessed by the coefficients of the corresponding characteristic polynomial. However, there are polynomials of degree higher than 2 with a structure similar to \eqref{P2} which also have roots with modulus different from 1. 
We prove here the following useful necessary condition on the structure of the characteristic polynomial of a $3 \times 3$ unitary matrix 
\begin{thm} \label{necP3_1}
The complex third degree characteristic polynomial of any unitary matrix of order 3 has necessarily the  structure:
\begin{eqnarray} \mathcal{P}_{3}(z)= z^3-re^{i\theta_1}z^2+re^{i(\theta-\theta_1)}z-e^{i\theta},
\label{upolychar}\end{eqnarray}
where $r\in [0,3]$, $\theta_1\in (-\pi,\pi]$ and $\theta\in (-\pi,\pi]$.
\end{thm}
\begin{proof}
Given any three real numbers $\alpha$, $\beta$, and $\gamma$, it is always possible to find three real numbers $r$, $\theta_1$ and $\theta$ that satisfy the following relations:
\begin{eqnarray} e^{i\alpha}+e^{i\beta}+e^{i\gamma}=re^{i\theta_1}, 
\label{uc1}\end{eqnarray}
\begin{eqnarray} e^{i\alpha}e^{i\beta}e^{i\gamma}=e^{i\theta}. 
\label{uc3}\end{eqnarray}
The product of \eqref{uc3} and the complex conjugate of \eqref{uc1} gives
\begin{eqnarray}re^{i(\theta-\theta_1)}= (e^{-i\alpha}+e^{-i\beta}+e^{-i\gamma})e^{i\alpha}e^{i\beta}e^{i\gamma}=e^{i(\alpha+\beta)}+e^{i(\beta+\gamma)}+e^{i(\alpha+\gamma)}, 
\label{uc2}\end{eqnarray}
where $r\in [0,3]$, $\theta_1\in (-\pi,\pi]$ and $\theta= \operatorname{Arg}e^{i(\alpha+\beta+\gamma)} \in (-\pi,\pi]$.
Equations \eqref{uc1} to \eqref{uc2} represent the Vieta-Girard formulas for the three roots $e^{i\alpha}$, $e^{i\beta}$, and $e^{i\gamma}$ of polynomial \eqref{upolychar}. Since these roots have modulus $1$, as is required for $\mathcal{P}_{3}(z)$ to be the characteristic polynomial of a unitary matrix, the Vieta-Girard formulas \eqref{uc1} and \eqref{uc3} clearly show that the complex coefficients of the characteristic polynomial of any $3\times 3$ unitary matrix are not independent. In fact, the free term and the coefficient of $z^2$, which are involved in Equations \eqref{uc1} and \eqref{uc3}, respectively, univocally determine the coefficient of $z$ through \eqref{uc2}, in accordance with \eqref{upolychar}. 
\end{proof}
Like Theorem \ref{Pn}, Theorem \ref{necP3_1} can be extended to a generic degree $n$. We emphasize that the polynomial form \eqref{upolychar} and its roots have some remarkable properties. For example, the passage from $z$ to the auxiliary variable $u=ze^{i\psi}$ leads, up to a global phase factor, to a polynomial with the same structure \eqref{upolychar} after the angle shifts $\theta_1'=\theta_1+\psi$ and  $\theta'=\theta+3\psi$ (these shifts are unimportant for what concerns the polynomial structure, since the angles $\theta_1'$ and $\theta'$ can take the same range of values as $\theta_1'$ and $\theta'$). Therefore, the three roots of the new polynomial have the same modules and relative principal arguments as the roots of the original polynomial \eqref{upolychar}. Another interesting property resulting from Equation \eqref{uc3} is that, if \eqref{upolychar} admits one root with modulus one, the other two roots must have reciprocal modules (including the case in which they also have modulus one).

Note that Theorem \ref{necP3_1} only expresses a necessary condition, and therefore there exist polynomials with the structure \eqref{upolychar} which do not admit a unitary ACM. Algebraic relations among the three parameters in the expression of $\mathcal{P}_{3}(z)$ not implied by the structure of the polynomial itself can ensure that $\mathcal{P}_{3}(z)$ admits a unitary ACM (vide infra).

Consider, for example, the case $r=3$. The polynomial $z^3-3z^2+3e^{i\pi}z-e^{i\pi}$, has the form \eqref{upolychar}, from which it is obtained by (arbitrarily) choosing $\theta_1=0$ and $\theta=\pi$. This is not the characteristic polynomial of a unitary $3\times 3$ matrix, since its roots are $-1$ and $2\pm\sqrt{3}$; accordingly, $\theta=\pi\neq\operatorname{Arg}e^{i(3\theta_1)}=0$. The relation $\theta=3\theta_1$ guaranties, instead, the existence of three roots of modulus $1$ when $r=3$, as is easily seen geometrically or from the fact that in this case $\mathcal{P}_{3}(z)=(z-e^{i\theta_1})^3$. As another example, for $r=1$, \eqref{upolychar} admits a unitary ACM for any $\theta_1\in (-\pi,\pi]$ and $\theta\in (-\pi,\pi]$, as the roots of the polynomial are $e^{i\theta_1}$ and $e^{i\frac{\theta-\theta_1\pm\pi}{2}}$.

The analysis in the next section will provide expressions for the coefficients of polynomial \eqref{upolychar} that make it the characteristic polynomial of a unitary ACM.

\subsection{Construction of a Trial Unitary ACM}
\par In Section \ref{sec:cubic}, it was convenient to search for an ACM of a generic monic complex polynomial \eqref{cubic} of third degree in the unknown $z$ by resorting to the canonical form \eqref{cubic canonical} of the polynomial through a translation of the complex variable $z$. The advantages of using the polynomial \eqref{cubic canonical} in the translated variable $\eta$ are twofold: the number of parameters appearing in the polynomial expression is reduced from $3$ to $2$ (namely, $p$ and $q$ instead of $c_1$,$c_2$, and $c_3$), and the very simple recipe \eqref{Q3Q3c} allows to obtain the ACM of the given polynomial from the ACM of its canonical form \eqref{cubic canonical}. 

It is clearly possible to pass from $\mathcal{P}_{3}(z)$ to its canonical form through the appropriate translation of $z$. Unfortunately such a strategy is not convenient in this case, since the canonical polynomial generally does not admit a unitary ACM, and thus the further mathematical step complicates the achievement of our goal. We therefore propose here a different approach that combines geometrical and analytical considerations. 

Exploiting Theorems \ref{Pn} and \ref{tp2}, we can represent each element $\mathbb{P}_{3}(z)$ of the set $[\mathbb{P}_{3}(z)]$ of all and only the third-degree polynomials that admit a unitary ACM and share the root $1$ as follows:
\begin{eqnarray}
\mathbb{P}_{3}(z) =(z^2-r_2e^{i\vartheta}z+e^{2i\vartheta})(z-1)=z^3-(1+r_2e^{i\vartheta})z^2+(1+r_2e^{-i\vartheta})e^{2i\vartheta}z-e^{2i\vartheta},
\label{P3R}\end{eqnarray}
where $r_2 \in[0,2]$ and $\vartheta\in (-\pi,\pi]$. Each polynomial \eqref{P3R} possesses, by construction, a unitary ACM and, vice versa, the characteristic polynomial of any $3 \times 3$ unitary  matrix with a unit eigenvalue is a particular realization of \eqref{P3R}.

$[\mathbb{P}_{3}(z)]$ is a subset of the set $[\mathcal{P}_{3}(z)]$ of the polynomials of the form \eqref{upolychar}. This point is appreciated by noting that the coefficient of $z^2$ in Equation \eqref{upolychar} can always be represented as
\begin{eqnarray}\label{c1}
re^{i\theta_1}= (1+r_2e^{i\vartheta})
\end{eqnarray}
with 
\begin{eqnarray}
{r_2}^2=1+r^2-2r\cos{\theta_1}
\label{r2}\end{eqnarray}
and 
\begin{eqnarray}
\vartheta=\operatorname{Arg}{\left(-1+r e^{i\theta_1}\right)} =2\operatorname{Arctan}{\left(\frac{r\sin\theta_1}{-1+r\cos\theta_1+ \sqrt{1+r^2-2r\cos\theta_1}}\right)}.
\label{alf}\end{eqnarray}
 The last equality is based on the following identity \cite{armitage2006elliptic}, which gives the principal argument of a generic complex number $(x+iy)\in \Omega$, where $\Omega$ coincides with the complex plane cut along the negative $x$-axis:
\begin{eqnarray}\operatorname{Arg} (x+iy)= 2\operatorname{Arctan}\left(\frac{y}{x+\sqrt{x^2+y^2}}\right).
\label{formulaArg_1}\end{eqnarray}
As expected, this formula leaves the argument of a complex number of null modulus undefined and, for any fixed negative $x$, implies   
\begin{eqnarray}\label{lim}
   \lim\limits_{y\longrightarrow 0^{\pm} }\operatorname{Arg}(x+iy) =\pm \pi.
\end{eqnarray} 
The above equations clearly show that $[\mathbb{P}_{3}(z)]$ is obtained as a subset of $[\mathcal{P}_{3}(z)]$ by introducing the relations \eqref{c1} to \eqref{alf} among the parameters $r$, $\theta_1$, and $\theta$, which are arbitrary in their ranges of definition in the polynomial expression \eqref{upolychar}. In particular, $\theta$ is compatible with \eqref{P3R} only if
\begin{eqnarray}\label{alfa}
\theta=2\vartheta,
\end{eqnarray}
thus leading to the following:
\begin{thm} \label{necP3}
A monic third-degree polynomial \eqref{upolychar} belongs to the set $[\mathbb{P}_{3}(z)]$ if and only if it can be written in the form 
\begin{eqnarray}
\tilde{\mathbb{P}}_{3}(z)= z^3-(1+r_2e^{i{\frac{\theta}{2}}})z^2+(1+r_2e^{-i{\frac{\theta}{2}}} )e^{i\theta}z-e^{i\theta},
\label{pteta}\end{eqnarray} 
where $r_2\in [0,2]$ and $\theta\in (-\pi,\pi]$.
\end{thm}
 
Note that the range of $r_2$ values is dictated by Theorem \ref{tp2}, and the corresponding range of $r$ values resulting from Equation \eqref{r2} for any given $\theta_1$ is a subset of the interval $[0,3]$ in Equation \eqref{upolychar}, in accordance with the fact that $[\mathbb{P}_{3}(z)]$ is a subset of $[\mathcal{P}_{3}(z)]$. On the other hand, since $\vartheta\in (-\pi,\pi]$, Equation \eqref{alfa} implies that $\theta\in (-2\pi,2\pi]$, which can be clearly reduced to the principal interval $[-\pi,\pi]$. We emphasize that requiring \eqref{upolychar} to be the characteristic polynomial of a unitary matrix with a real positive eigenvalue (that is, 1) entailed relations between the three parameters in Equation \eqref{upolychar}, thus leading to the dependence of the polynomial \eqref{pteta} on only two parameters, $r_2$ and $\theta$.

At this point, we can construct an ACM for a polynomial of the kind \eqref{pteta}. The polynomial factorization in Equation \eqref{P3R} enables a block diagonal form for the ACM, with the one-dimensional block simply equal to 1. The diagonal elements of the $2\times2$ block can be set equal \cite{Horn} and are immediately obtained from Vieta's formula for the sum of the roots of polynomial \eqref{P2}. Then, simple algebraic considerations lead to the following ACM of polynomial \eqref{pteta}:
\begin{equation}\label{Wtilde}
    \widetilde{W}_{3}= \left(
                     \begin{array}{ccc} 
  \frac{r_2}{2}e^{i\frac{\theta}{2}}\ &\sqrt{1-\left( \frac{r_2}{2}\right)^2}e^{i\frac{\theta}{2}}\ &0\\
  -\sqrt{1-\left(\frac{r_2}{2}\right)^2}e^{i\frac{\theta}{2}} \ & \frac{r_2}{2}e^{i\frac{\theta}{2}}\   &0\\
    0&0&  1\\
   \end{array}
     \right). \end{equation}
It is easy to verify that the characteristic polynomial of \eqref{Wtilde} coincides with \eqref{pteta} for all the allowed values of the parameters $r_2$ and $\theta$. The three columns of $\widetilde{W}_{3}$ are normalized and mutually orthogonal, and these properties imply the unitarity of matrix $\widetilde{W}_{3}$. If $ r_2>2$, and thus it is out of the range given in Theorem \ref{necP3}, the first two columns of the matrix are not orthogonal for any $\theta$, and then $\widetilde{W}_{3}$ is no longer unitary.

It is worth noting that \eqref{Wtilde} can be seen as the result of a partial diagonalization of another unitary matrix with the same eigenvalues and that all the other ACMs of a polynomial \eqref{pteta} can be obtained by unitary transformation of \eqref{Wtilde}.

Once an ACM is constructed for any polynomial \eqref{pteta}, the subset of $[\mathcal{P}_{3}(z)]$ that contains all and only the polynomials \eqref{upolychar} admitting a unitary ACM can be generated by rotating the roots of each polynomial \eqref{pteta} by an angle $\epsilon\in (-\pi,\pi]$. This amounts to changing the complex variable $z$ to $u=ze^{i\epsilon}$ in $\tilde{\mathbb{P}}_{3}(z)$. Then, up to a global phase factor $e^{3i\epsilon}$, the polynomial $\tilde{\mathbb{P}}_{3}(u)=\tilde{\mathbb{P}}_{3}(ze^{i\epsilon})$ is equal to
\begin{eqnarray}
\mathbb{P}_{3\epsilon}(z)=
z^3-(1+r_2e^{i{\frac{\theta}{2}}})e^{-i\epsilon}z^2+(1+r_2e^{-i{\frac{\theta}{2}}} )e^{i\theta}e^{-2i\epsilon}z-e^{i(\theta-3\epsilon)}.
\label{3RR}\end{eqnarray}
The root 1 of $\tilde{\mathbb{P}}_{3}(z)$ corresponds to a general complex root of modulus one, $e^{i\epsilon}$, of $\mathbb{P}_{3\epsilon}(z)$.
$[\mathbb{P}_{3\epsilon}(z)]$ includes all and only the polynomials $\mathbb{P}_{3\epsilon}(z)$ which, by construction, admit a unitary ACM and therefore also satisfy the necessary condition  expressed by Theorem \ref{necP3}. 
We have thus proved that
\begin{thm} \label{tp3}
A complex monic polynomial of third degree is the characteristic polynomial of a unitary matrix of order 3 if and only if it has the structure \eqref{3RR},  
with $r_2\in [0,2]$, $\theta\in (-\pi,\pi]$, and  $\epsilon\in (-\pi,\pi]$.
\end{thm}

Next, we accomplish the main objective of this section by constructing an ACM $W_3$ of $\mathbb{P}_{3\epsilon}(z)$. To this end, in analogy with the recipe \eqref{Q3Q3c}, we use the transformation $z=ue^{-i\epsilon}$ to generate $W_3$ from the matrix $\widetilde{W}_{3}$ in Equation \eqref{Q3Q3c} by proceeding as follows:
\begin{eqnarray}\tilde{\mathbb{P}}_{3}(u)
 &\equiv& det(u I_3-\widetilde{W}_{3})=det(e^{i\epsilon}(z I_3-e^{-i\epsilon}\widetilde{W}_{3} ))\nonumber\\&=& e^{3i\epsilon}det(zI_3- {W}_{3}) \equiv e^{3i\epsilon} \mathbb{P}_{3\epsilon}(z),\label{WWtilde}\end{eqnarray}
 where $W_3\equiv e^{-i\epsilon}\widetilde{W}_{3}$.
In light of our previous arguments, the characteristic polynomial of the matrix 
\begin{equation}\label{WW}
W_3=e^{-i\epsilon}   \widetilde{W}_{3}= \left(
                     \begin{array}{ccc} 
  \frac{r_2}{2}e^{i(\frac{\theta}{2}-\epsilon})\ &\sqrt{1-\left( \frac{r_2}{2}\right)^2}e^{i(\frac{\theta}{2}-\epsilon})\ &0\\
  -\sqrt{1-\left(\frac{r_2}{2}\right)^2}e^{i(\frac{\theta}{2}-\epsilon})  \ & \frac{r_2}{2}e^{i(\frac{\theta}{2}-\epsilon}) \   &0\\
    0&0& e^{-i\epsilon} \\
   \end{array}
     \right),  \end{equation} 
is the polynomial \eqref{3RR}. The unitary matrix \eqref{WW} fully  meets our goal. To reach this objective,  we first characterized the class $[\mathbb{P}_{3\epsilon}(z)]$ of all and only the polynomials that admit a unitary ACM, thus removing the difficulties related to the lack of sufficiency of polynomial \eqref{upolychar}. Then, exploiting the recipe \eqref{WWtilde},  we established the form $W_{3}$ of a unitary ACM  for any polynomial belonging to $[\mathbb{P}_{3\epsilon}(z)]$.

In the following, we illustrate our approach through some applications.
\subsection{Examples}
\begin{itemize}
\item Given the parameter $\sigma=\pm 1$, consider the subclass of polynomials  \eqref{3RR} with $\epsilon= \frac{(1-\sigma)\pi}{2}$: \begin{eqnarray}
\mathbb{P}_{3\sigma'}(z)&=&
z^3-(1+r_2e^{i{\frac{\theta}{2}}})e^{-i\frac{(1-\sigma)\pi}{2}}z^2+(1+r_2e^{-i{\frac{\theta}{2}}} )e^{i\theta} z-e^{i(\theta- \frac{(1-\sigma)\pi}{2})}\nonumber\\&=&z^3-(1+r_2e^{i{\frac{\theta}{2}}}) \sigma z^2+(1+r_2e^{-i{\frac{\theta}{2}}} )e^{i\theta} z-\sigma e^{i\theta}, 
\label{3sigma'}\end{eqnarray} 
In Equation \eqref{3sigma'} we exploited the identity $e^{-i\frac{(1-\sigma)\pi}{2}}=\sigma$. It is easy to verify that $\mathbb{P}_{3\sigma}(\sigma)=0$, which means that \eqref{3sigma'} is the most general polynomial with the real root $\sigma$ and an ACM which, in view of  \eqref{WW}, can be written on demand as
\begin{eqnarray}
W_{3\sigma}=\left(
                     \begin{array}{ccc} 
  \sigma\frac{r_2}{2}e^{i\frac{\theta}{2}}\ &\sigma\sqrt{1-\left( \frac{r_2}{2}\right)^2}e^{i\frac{\theta}{2}}\ &0\\
  -\sigma\sqrt{1-\left(\frac{r_2}{2}\right)^2}e^{i\frac{\theta}{2}} \ & \sigma\frac{r_2}{2}e^{i\frac{\theta}{2}}\   &0\\
    0&0& \sigma \\
   \end{array}
     \right).  \label{Wsigma}\end{eqnarray}  
\item
Consider the polynomial
\begin{eqnarray}
\mathbb{P}_{3,r_2=0}(z)&=&
z^3-e^{i\theta_1}z^2+ e^{i\theta'}e^{-i\theta_1} z-e^{i\theta'}, 
\label{r2=0}\end{eqnarray}
obtained by setting $r_2=0$, $\epsilon=-\theta_1$, and $\theta=\theta'-3\theta_1$ in Equation \eqref{3RR}. With these choices, using \eqref{WW} we immediately find that the polynomial \eqref{r2=0} admits the ACM
\begin{eqnarray}
W_{r_2=0}=\left(
                     \begin{array}{ccc} 
   0\ & e^{i\frac{\theta'-\theta_1}{2}}\ &0\\
  - e^{i\frac{\theta'-\theta_1}{2}} \ &  0\   &0\\
    0&0&  e^{i\theta_1} \\
   \end{array}
     \right).  \label{Wr_2=0}\end{eqnarray}
 
When $r=1$, the polynomial \eqref{r2=0} coincides with the polynomial \eqref{upolychar} up to a trivial change of notation ($\theta$ is substituted with $\theta'$). This polynomial is then the characteristic polynomial of a unitary matrix for any $\theta'$ and $\theta_1$, as we observed in Subsection \ref{U41}, with eigenvalues that are now immediately derived from the matrix \eqref{Wr_2=0} as $e^{i\theta_1}$ and $e^{i\frac{\theta'-\theta_1\pm\pi}{2}}$.
 \item 
For $r= 0$, \eqref{upolychar} yields the polynomial $z^3-e^{i\theta}=0$ and, whatever the $\epsilon$ value, Equations \eqref{c1} and \eqref{alfa} imply that $r_2=1$, and $\theta=2\pi$ in Equation \eqref{3RR}. The corresponding ACM, with structure \eqref{WW}, takes the form
  \begin{eqnarray}
W_{r=0}=\left(
                     \begin{array}{ccc} 
  \frac{1}{2}e^{i(\pi-\epsilon)}\ &\sqrt{  \frac{3}{4}}e^{i(\pi-\epsilon)}\ &0\\
  -\sqrt{  \frac{3}{4}}e^{i(\pi-\epsilon)}\ & \frac{1}{2}e^{i(\pi-\epsilon)}  \ &0\\
    0&0&  e^{-i\epsilon}  \\
   \end{array}
     \right),  \label{Wsigma_1}\end{eqnarray}
Apart from the cases $r\neq 1$ or $r\neq 0$ examined above, for a generic $r$ Equation \eqref{upolychar} admits a unitary ACM only under $r$-dependent algebraic constraints on $\theta$ and $\theta_1$. These constraints are realized in the polynomial form \eqref{3RR} through Equations \eqref{c1} to \eqref{alf} and \eqref{alfa}, for the ranges of parameter values defined in Theorem \ref{tp3}.
 \item
These conditions are not all satisfied by the polynomial 
\begin{eqnarray}
\mathbb{P}_{3,r=2}(z)&=&
z^3-2e^{i \pi}e^{-i\epsilon}z^2+ 2e^{i \pi}e^{-2i\epsilon}  z-e^{-3i\epsilon},
\label{r=2}\end{eqnarray} 
which is of the form \eqref{upolychar} with $r=2$, $\theta_1= -\epsilon +\pi$ and $\theta=2\pi-3\epsilon$. It is easy to see that this polynomial coincides with 
\begin{eqnarray}
\mathbb{P}_{3,r=2}(z)&=&
z^3-(1+3e^{i \pi})e^{-i\epsilon}z^2+ (1+3e^{-i \pi})  e^{-2i\epsilon}z-e^{-3i\epsilon}, 
\label{alternativer=2}\end{eqnarray}
which has the form \eqref{3RR} with $\theta=2\pi$, except for the fact that $r_2\notin [0,2]$, as is instead required in Theorem \ref{tp3} because of Theorem \ref{tp2}. As a consequence, the polynomial \eqref{alternativer=2} or, equivalently, \eqref{r=2} cannot be the characteristic polynomial of a unitary $3\times 3$ matrix.
\item
Polynomial \eqref{r=2} obviously admits an ACM of the form \eqref{Q3_1}, which is obtained using the general protocol in Section \ref{sec:cubic}. Moreover, the insertion of $r_2=3$ and $\theta= 2\pi$ in matrix \eqref{WW} leads to another ACM of polynomial \eqref{r=2} of the form
\begin{equation}\label{WWNU}
W_{3nu}=  \left(
                     \begin{array}{ccc} 
  -\frac{3}{2}e^{-i\epsilon}\ &-\frac{i}{2}\sqrt{5}e^{-i\epsilon}\ &0\\
   \frac{i}{2}\sqrt{5}e^{-i\epsilon}  \ & -\frac{3}{2}e^{-i\epsilon} \   &0\\
    0&0& e^{-i\epsilon} \\
   \end{array}
     \right),  \end{equation} 
which is manifestly non unitary.
In general, $W_{3nu}$ and the ACM of form \eqref{Q3_1} are not similar. A sufficient condition for their similarity is that their three common complex eigenvalues are distinct. In this case, in fact, both matrices are surely diagonalizable \cite{Horn} and therefore traceable (in general, not unitarily) to the same diagonal matrix.
\end{itemize}

\section{Discussion and Conclusions }\label{discussion_conclusion}
Finding a matrix with an assigned characteristic polynomial is a classic inverse problem solved by Frobenius a long time ago. In Section ~\ref{inverse problem} we observed that there are infinitely many other solutions, generally not equivalent to the one found by Frobenius. The exhaustive description of the set $S(P_n(z))$ of all complex matrices sharing the same characteristic polynomial is still an open problem. Among the reasons for the missing solution of this problem it is useful to consider that, even if $S(P_n(z))$ is invariant under similarity transformations, it includes non-similar (and in general not even equivalent) matrices. Another problem is, for example, that the structure of the non-empty subset of non-sparse CMs belonging to $S(P_n(z))$, unlike the set of sparse CMs \cite{deaett2019non}, has not yet been fully characterized.  
 
A related inverse problem, stimulated by applications of current interest to both physicists and mathematicians, is the search for ACMs \textit{constrained} to possess prescribed structural properties, such as, for example, unitarity,  positive semidefiniteness, or Hermiticity.

The first focus of this study is the construction of a new ACM of a generic monic and complex third degree polynomial $P_3(z)$, characterized by versatility for applications. This objective is pursued by parameterizing the elements of the ACM in such a way that they lend themselves to additional constraints dictated by specific problems (of which only the structural properties are exploited).
To the best of our knowledge, our investigation of inverse problems of this kind opens a new fruitful chapter in this research area whose central goal is the proposal of new ACMs which, in particular, can be CMs. We address the three above-mentioned constrained inverse problems, providing methodology and results that aim for broad applicability and are potentially transferable to solving analogous problems involving higher degree polynomials. 

\par The adopted step-by-step approach builds new specific classes of unconstrained or constrained matrices as ACMs of suitably given polynomials, relaxing from the outset any condition on the degree of the minimal polynomial. In particular, the strategy implemented, as well as the mathematical tools used, is not influenced by any FCM-based or FCM-inspired technique.
\par The elements of the ACM that we construct as a solution of the general inverse problem are single-valued complex functions of the coefficients of the given generic polynomial. Exploiting the structural properties of this ACM, we find the algebraic expressions for the three, generally complex, roots of its complex characteristic polynomial.

It is remarkable that, when the polynomial becomes real, the associated general ACM smoothly becomes Hermitian under easy to find necessary and sufficient conditions on the coefficients of the polynomial described by \eqref{1629}. Using simply the fact that the ACM becomes Hermitian if and only if $\Delta\leq 0$, we are able to extend the trigonometric representation of the three roots of the real polynomial to all possible cases, that is even when \eqref{deltanonpositivo} does not hold. This representation is obtained without resorting to the well-known Cardano-Del Ferro formulas.

\par We emphasize that the FCM of a characteristic polynomial does not undergo any structural change when the coefficients of the complex polynomial become real, thus providing no additional information on the polynomial roots. For this reason, we claim that the FCM has a lower flexibility than our ACM. 
We show how to use this flexibility to obtain an ACM of a prescribed characteristic polynomial on demand, applying our procedure to the important problem of finding a density matrix and, in particular, that of a qutrit.  

\par A second, novel constrained inverse problem addressed here consists in finding a unitary ACM of a generic polynomial with three roots of modulus one.
Excluding the trivial case in which the unitary ACM can be directly given in diagonal form, we first reach the intermediate goal (interesting in itself) of parameterizing $c_1$, $c_2$, and $c_3$ in the polynomial \eqref{cubic} so as to set the necessary conditions on the structure of a polynomial to admit an ACM. By setting appropriate relations between the parameters through coupled geometric and analytical considerations, we further constrain the polynomial structure in such a way as to identify the set $[\mathbb{P}_{3\epsilon}(z)]$ of all and only the third-degree polynomials that admit a unitary ACM. Then, we conclude our analysis by constructing the associated ACM.

\par The results of this study can be further explored and usefully applied to physical and  mathematical contexts, including, at their intersection, the research area of quantum computing. For example, for time-dependent parameters, the prescription of a time-dependent characteristic polynomial of $[\mathbb{P}_{3}(z)]$ or $[\mathbb{P}_{3\epsilon}(z)]$ leads to a unitary time-dependent matrix, hence to the pertinent time-dependent Hamiltonian that generates the time evolution of a qutrit, whose properties can thus be traced back to the prescription of $\mathbb{P}{_3}[z]$.
In mathematical contexts, the analysis developed in this study can be extended to polynomials of higher degree. 
For example, in the case of a fifth-degree polynomial, our protocol may provide conditions on the coefficients of the polynomial such that its roots are real.
The rich analysis enabled by the use of third-degree polynomials in this study sets a clearer basis to conceive the extension of the analysis to higher degree polynomials.

\section*{Funding}
L.M. acknowledge funding from the NWO Gravitation Program Quantum Software Consortium. A.M. (Agostino Migliore) acknowledges funding from the European Union - NextGenerationEU, within the National Center for HPC, Big Data and Quantum Computing (Project no. CN00000013, CN1 Spoke 10: "Quantum Computing").
\section*{Acknowledgement}
L.M. thanks Prof. Maria Carmela Lombardo for her invitation at the University of Palermo for scientific collaboration.

\appendix
\section[\appendixname~\thesection]{Derivation of Equations \eqref{rootspg01} to \eqref{rootspl04}}
\numberwithin{equation}{section}
\renewcommand{\theequation}{\thesection\arabic{equation}}
In this appendix, we find the roots of the polynomial \eqref{cubic} in the case in which
\begin{eqnarray}
	\Delta(p,q)\equiv\frac{p^3}{27}+\frac{q^2}{4}\geq 0,
	\label{A1}
\end{eqnarray}
From Equation \eqref{FI13real} it is immediately seen that the inequality \eqref{A1} implies $\chi^2<1$ or $\chi^2>1$ depending on whether $p>0$ or $p<0$, respectively. As a consequence, Equations \eqref{FI13}, \eqref{magicroot}, and \eqref{etapmroot} lead to different polynomial roots depending on the sign of $p$, in agreement with previous analyses \cite{Tricomi}.

We begin with considering the case $p>0$, which means $\Theta_p=0$, whence $\varphi_p=\pi$ according to Equation \eqref{pefi} and $\chi$ is given by Equation \eqref{rootspg06}, that is,
\begin{equation} \chi=-iu, \quad u=\frac{3q}{2p}\sqrt{\frac{3}{p}}. \end{equation}
In this case, $\operatorname{Arg}(1-\chi^2)=0$ and Equation \eqref{FI13} gives
\begin{eqnarray}
	\Phi_{13}&=&\operatorname{Arg}\left(\chi+i\sqrt{1-\chi^2}\right)
	-i\operatorname{Ln}\left(\chi+ i\sqrt{1-\chi^2}\right) \nonumber\\ 
	&=&\frac{\pi}{2}-i\operatorname{ln}\left(\sqrt{1+u^2}-u\right)
	=\frac{\pi}{2}+i\operatorname{ln}\left(\sqrt{1+u^2}+u\right).
\end{eqnarray}
Then, using Euler's formula and rationalization, Equation \eqref{magicroot} easily yields
\begin{equation}
\tilde\eta_1=\frac{\sqrt{3}}{2}(A+B)+\frac{i}{2}(A-B)
\end{equation}
with
\begin{eqnarray}
	A=\sqrt[3]{\sqrt{1+u^2}-u}, \quad B=\sqrt[3]{\sqrt{1+u^2}+u},
\end{eqnarray}
whence
\begin{eqnarray}
	z_1=i\sqrt{\frac{p}{3}}\tilde\eta_{1}-\frac{c_1}{3}
	=\frac{\sqrt{p}}{2}\left[\frac{1}{\sqrt{3}}(B-A)+i(A+B)\right]-\frac{c_1}{3},
\end{eqnarray}
namely Equation \eqref{rootspg01}, being $u=i\chi$ and
\begin{eqnarray}
	X=A+B, \quad Y=\frac{B-A}{\sqrt{3}},
\end{eqnarray}
Considering that $AB=1$ and using Equation \eqref{etapmroot}, after some lengthy algebra one obtains
\begin{eqnarray}
	\tilde\eta_{k}&=&\frac{\sqrt{3}}{2}\left[-\frac{A+B}{2}+i\frac{B-A}{2\sqrt{3}}
	+(-1)^{k+1}\sqrt{1-\frac{(B-A)^2}{2}+i\frac{\sqrt{3}(B^2-A^2)}{2}}\right]
	\nonumber\\&=&\frac{\sqrt{3}}{2}\left[-\frac{X}{2}+i\frac{Y}{2}
	+\frac{(-1)^{k+1}}{2}\sqrt{(X+3iY)^2}\right], \quad k=2,3.
\end{eqnarray}
Choosing the root in the last expression in accordance with Equation \eqref{etapmroot}, we thus have
\begin{equation} \tilde\eta_{2}=-\frac{\sqrt{3}}{2}(X+iY) \end{equation}
and
\begin{equation} \tilde\eta_{3}=i\sqrt{3}Y, \end{equation}
from which the root expressions \eqref{rootspg02} and \eqref{rootspg03} immediately result.

Consider, for example, the polynomial $z^3+4z-7\sqrt{3}=0$, which is directly given in canonical form, so that $c_1=0$, $c_2=p$, and $c_3=q$. In this case, $u=-7\frac{9}{16}$, whence $A=2$, $B=\frac{1}{2}$ and therefore $X=\frac{5}{2}$, $Y=-\frac{\sqrt{3}}{2}$. Equations \eqref{rootspg01} to \eqref{rootspg03} thus give 
$z_1=\frac{1}{2}(-\sqrt{3}+5i)$,$z_2=-\frac{1}{2}(\sqrt{3}+5i)$, and $z_3=\sqrt{3}$, which are readily verified to be the roots of the above polynomial.

\par For $p<0$, we have $\Theta_p=0$, from which $\varphi_p=2\pi$, and
\begin{equation} \chi=-\frac{3q}{2p}\sqrt{-\frac{3}{p}}. \end{equation}
Therefore, $\operatorname{Arg}(1-\chi^2)=\pi$ and Equation \eqref{FI13} gives
\begin{eqnarray}
	\Phi_{13}&=&\operatorname{Arg}\left(\chi-\sqrt{\chi^2-1}\right)
	-i\operatorname{ln}\left(\chi-\sqrt{\chi^2-1}\right) \nonumber\\ 
	&=&i\operatorname{ln}\left(\chi+\sqrt{\chi^2-1}\right)
	\equiv i\nu.
\end{eqnarray}
Using Euler's formula again, we obtain
\begin{eqnarray}
\tilde\eta_1=2\cos\left(i\frac{\nu}{3}\right)
=2\cosh\left(\frac{\nu}{3}\right)=C
\end{eqnarray}
with $C$ given by Equation \eqref{rootspl03}, from which Equation \eqref{rootspl01} immediately follows.

Since
\begin{eqnarray}
\sin^2(i\frac{\nu}{3})=-\sin^2(\frac{\nu}{3})
=1-\cosh^2(\frac{\nu}{3})=1-\frac{C^2}{4}<0 \quad \forall \nu,
\end{eqnarray}
Equation \eqref{etapmroot} gives
\begin{eqnarray}
	\tilde\eta_{k}&=&-\left[\frac{C}{2}
	+i(-1)^k\sqrt{3\left(\frac{C^2}{4}-1\right)}\right], \quad k=2,3,
\end{eqnarray}
whence Equation \eqref{rootspl02}.

\bibliographystyle{unsrt} 
\bibliography{fer_12}

\begin{thebibliography}{10}

\bibitem{Frobenius}
G.~Frobenius.
\newblock Theorie der linearen formen mit ganzen coefficienten.
\newblock {\em J. Reine Angew. Math.}, 86:146--208, 1879.

\bibitem{Horn}
R.~A. Horn and C.~R. Johnson.
\newblock {\em Matrix Analysis}.
\newblock Cambridge University Press, 2013.

\bibitem{Hawkins}
T.~Hawkins.
\newblock {\em The mathematics of {F}robenius in context: A journey through
  18th to 20th century mathematics, Sources and Studies in the History of
  Mathematics and Physical Sciences}.
\newblock Springer, New-York, 2013.

\bibitem{Loewy_1917}
A.~Loewy.
\newblock Die begleitmatrix eines linearen homogenen differentialstatusdruckes.
\newblock {\em Nachr. Ges. Wiss. Gottingen, Math.-Phys. Kl}, 1:255--263, 1917.

\bibitem{Barnett_2}
S.~Barnett.
\newblock Congenial matrices.
\newblock {\em Linear Algebra Appl.}, 41:277--298, 1981.

\bibitem{doi:10.1137/120865392}
J.~L. Aurentz, R.~Vandebril, and D.~S. Watkins.
\newblock Fast computation of the zeros of a polynomial via factorization of
  the companion matrix.
\newblock {\em SIAM J. Scien. Comp.}, 35(1):A255--A269, 2013.

\bibitem{DETERAN2014197}
F.~{De Teran}, F.~M. Dopico, and J.~Perez.
\newblock New bounds for roots of polynomials based on fiedler companion
  matrices.
\newblock {\em Linear Algebra Appl.}, 451:197--230, 2014.

\bibitem{bini1996numerical}
Dario~Andrea Bini.
\newblock Numerical computation of polynomial zeros by means of aberth's
  method.
\newblock {\em Num. alg.}, 13(2):179--200, 1996.

\bibitem{HIGHAM20035}
Nicholas~J. Higham and Francoise Tisseur.
\newblock Bounds for eigenvalues of matrix polynomials.
\newblock {\em Linear Algebra Appl.}, 358(1):5--22, 2003.

\bibitem{bueno2011recovery}
M.~I. Bueno, F.~De~Ter{\'a}n, and F.~M. Dopico.
\newblock Recovery of eigenvectors and minimal bases of matrix polynomials from
  generalized fiedler linearizations.
\newblock {\em SIAM J. Matrix Anal. Appl.}, 32(2):463--483, 2011.

\bibitem{antoniou2004new}
E.~Antoniou and S.~Vologiannidis.
\newblock A new family of companion forms of polynomial matrices.
\newblock {\em Electron. J. Linear Algebra}, 11:78--87, 2004.

\bibitem{brand1968applications}
L.~Brand.
\newblock Applications of the companion matrix.
\newblock {\em Amer. Math. Monthly}, 75(2):146--152, 1968.

\bibitem{wardlaw1994matrix}
William~P Wardlaw.
\newblock Matrix representation of finite fields.
\newblock {\em Math. Mag.}, 67(4):289--293, 1994.

\bibitem{szederkenyi2006intelligent}
G{\'a}bor Szederk{\'e}nyi, Roz{\'a}lia Lakner, and Mikl{\'o}s Gerzson.
\newblock {\em Intelligent control systems: an introduction with examples},
  volume~60.
\newblock Springer Science \& Business Media, 2006.

\bibitem{LIM20112921}
A.~Lim and J.~Dai.
\newblock On product of companion matrices.
\newblock {\em Linear Algebra Appl.}, 435(11):2921--2935, 2011.

\bibitem{Specht_3}
W.~Specht.
\newblock Die lage der nullstellen eines polynoms iii.
\newblock {\em Math. Nachr.}, 16:257--263, 1957.

\bibitem{Specht_4}
W.~Specht.
\newblock Die lage der nullstellen eines polynoms iv.
\newblock {\em Math. Nachr.}, 21:201--222, 1960.

\bibitem{Good_1961}
I.J. Good.
\newblock The colleague matrix, a {C}hebyshev analogue of the companion matrix.
\newblock {\em Quart. J. Math. Oxford Ser.}, 12 (2):61--68, 1961.

\bibitem{maroulas}
J.~Maroulas and S.~Barnett.
\newblock Polynomials with respect to a general basis.
\newblock {\em J. Math. Anal. Appl}, 72(1):177--194, 1979.

\bibitem{Branden}
P.~Br{\"a}nd{\'e}n.
\newblock On linear transformations preserving the {P}olya frequency property.
\newblock {\em Rans. Amer. Math. Soc.}, 358:3697--3716, 2006.

\bibitem{HAGLUND20001017}
J.~Haglund.
\newblock Further investigations involving {R}ook polynomials with only real
  zeros.
\newblock {\em Eur. J. Comb.}, 21(8):1017 -- 1037, 2000.

\bibitem{PITMAN1997279}
J.~Pitman.
\newblock Probabilistic bounds on the coefficients of polynomials with only
  real zeros.
\newblock {\em J. Comb. Theory Ser. A.}, 77(2):279 -- 303, 1997.

\bibitem{WAGNER1991138}
D.~G. Wagner.
\newblock The partition polynomial of a finite set system.
\newblock {\em J. Comb. Theory Ser. A.}, 56(1):138 -- 159, 1991.

\bibitem{WAGNER1992459}
D.~G. Wagner.
\newblock Total positivity of {H}adamard products.
\newblock {\em J. Math. Anal. Appl.}, 163(2):459 -- 483, 1992.

\bibitem{WANG200563}
Y.~Wang and Y.-N. Yeh.
\newblock Polynomials with real zeros and {P}olya frequency sequences.
\newblock {\em J. Comb. Theory Ser. A.}, 109(1):63 -- 74, 2005.

\bibitem{7016940}
G.~Sun, S.~Su, and M.~Xu.
\newblock Quantum algorithm for polynomial root finding problem.
\newblock In {\em 2014 Tenth International Conference on Computational
  Intelligence and Security}, pages 469--473, 2014.

\bibitem{Nagata2018}
K.~Nagata, T.~Nakamura, H.~Geurdes, J.~Batle, A.~Farouk, D.N. Diep, and S.~K.
  Patro.
\newblock Efficient quantum algorithms of finding the roots of a polynomial
  function.
\newblock {\em Int. J. Theor. Phys.}, 57(8):2546--2555, Aug 2018.

\bibitem{Nagata2019}
K.~Nagata and T.~Nakamura.
\newblock Quantum algorithm for the root-finding problem.
\newblock {\em Quant. Stud.: Math. and Found.}, 1(6):2196--5617, 2019.

\bibitem{Tansuwannon2019}
T.~Tansuwannont, S.~Limkumnerd, S.~Suwanna, and P.~Kalasuwan.
\newblock Quantum phase estimation algorithm for finding polynomial roots.
\newblock {\em Open Phys. J.}, 17(1):839--849, 2019.

\bibitem{TAN200275}
L.~Tan and A.C. Pugh.
\newblock Spectral structures of the generalized companion form and
  applications.
\newblock {\em Syst. Control. Lett.}, 46(2):75--84, 2002.

\bibitem{Weigert_2003}
S.~Weigert.
\newblock A quantum search for zeros of polynomials.
\newblock {\em J. opt., B Quantum semiclass. opt.}, 5(6):S586--S588, oct 2003.

\bibitem{SCHMEISSER199311}
G.~Schmeisser.
\newblock A real symmetric tridiagonal matrix with a given characteristic
  polynomial.
\newblock {\em Linear Algebra Appl.}, 193:11--18, 1993.

\bibitem{FIEDLER1990265}
M.~Fiedler.
\newblock Expressing a polynomial as the characteristic polynomial of a
  symmetric matrix.
\newblock {\em Linear Algebra Appl.}, 141:265--270, 1990.

\bibitem{EASTMAN2014255}
B.~Eastman, I.-J. Kim, B.L. Shader, and K.N. {Vander Meulen}.
\newblock Companion matrix patterns.
\newblock {\em Linear Algebra Appl.}, 463:255--272, 2014.

\bibitem{deaett2019non}
Louis Deaett, Jonathan Fischer, Colin Garnett, and Kevin Vander~Meulen.
\newblock Non-sparse companion matrices.
\newblock {\em Electron. J. Linear Algebra}, 35:223--247, 2019.

\bibitem{borisenko1968vector}
A.I. Borisenko and I.~E. Tarapov.
\newblock {\em Vector and tensor analysis with applications}.
\newblock Courier Corporation, 1968.

\bibitem{kalman2000matrix}
D.~Kalman.
\newblock A matrix proof of newton's identities.
\newblock {\em Math. Mag.}, 73(4):313--315, 2000.

\bibitem{prasolov1994problems}
V.~V. Prasolov.
\newblock {\em Problems and theorems in linear algebra}, volume 134.
\newblock American Math. Soc., 1994.

\bibitem{boas2006mathematical}
Mary~L Boas.
\newblock {\em Mathematical methods in the physical sciences}.
\newblock John Wiley \& Sons, 2006.

\bibitem{hadamard1902problemes}
J.~Hadamard.
\newblock Sur les probl{\`e}mes aux d{\'e}riv{\'e}es partielles et leur
  signification physique.
\newblock {\em Princeton university bulletin}, pages 49--52, 1902.

\bibitem{kabanikhin2008definitions}
S.I. Kabanikhin.
\newblock Definitions and examples of inverse and ill-posed problems.
\newblock {\em J. Inv. Ill-Posed Problems}, 16:317--357, 2008.

\bibitem{von2011ill}
I.~von W{\"u}rtemberg.
\newblock Ill-posed problems and their applications to climate research,
  u.u.d.m. project report.
\newblock Technical report, Uppsala university, Department of Mathematics, June
  2011.

\bibitem{Barnett}
S.~Barnett.
\newblock {\em Matrices: Methods and Applications}.
\newblock Oxford University Press, 1990.

\bibitem{10.2307/2299273}
H.~G. Funkhouser.
\newblock A short account of the history of symmetric functions of roots of
  equations.
\newblock {\em Am. Math. Mon.}, 37(7):357--365, 1930.

\bibitem{abramowitz1964handbook}
M.~Abramowitz and I.A. Stegun.
\newblock {\em Handbook of mathematical functions with formulas, graphs, and
  mathematical tables}, volume~55.
\newblock US Government printing office, 1964.

\bibitem{haber2011complex}
{Haber, H.E.}
\newblock The complex inverse trigonometric and hyperbolic functions.
\newblock scipp.ucsc.edu/~haber/webpage/arc3.pdf, 2022.
\newblock University of California.

\bibitem{kahan1987branch}
W.~Kahan.
\newblock Branch cuts for complex elementary functions.
\newblock {\em The State of the Art m Numerical Analyszs, MJD Powell and A
  Iserles, Eds., oxford University Press, NY}, 1987.

\bibitem{doi:10.1142/1284}
G.~V. Milovanovic, D.~S. Mitrinovic, and Th.~M. Rassias.
\newblock {\em Topics in Polynomials}.
\newblock World Scientific, 1994.

\bibitem{Lambert}
W.~D. Lambert.
\newblock A generalized trigonometric solution of the cubic equation.
\newblock {\em Amer. Math. Monthly}, 13:73--76, 1906.

\bibitem{Dickson}
L.~E. Dickson.
\newblock {\em First course in the theory of equations}.
\newblock John Wiley \& Sons, Inc., 1922.

\bibitem{Tricomi}
F.~Tricomi.
\newblock {\em Lezioni di analisi matematica, Vol.1}.
\newblock CEDAM, 1965.

\bibitem{doi:10.1080/00029890.1992.11995845}
D.~C. Kurtz.
\newblock A sufficient condition for all the roots of a polynomial to be real.
\newblock {\em Am. Math. Mon.}, 99(3):259--263, 1992.

\bibitem{carmichael2001directions}
H.~J. Carmichael, R.~J. Glauber, and M.~O. Scully.
\newblock {\em Directions in Quantum Optics: A Collection of Papers Dedicated
  to the Memory of Dan Walls Including Papers Presented at the TAMU-ONR
  Workshop Held at Jackson, Wyoming, USA, 26--30 July 1999}, volume 561.
\newblock Springer Science \& Business Media, 2001.

\bibitem{fano1957description}
U.~Fano.
\newblock Description of states in quantum mechanics by density matrix and
  operator techniques.
\newblock {\em Rev. Mod. P.}, 29(1):74, 1957.

\bibitem{blanchard2015mathematical}
P.~Blanchard and E.~Br{\"u}ning.
\newblock {\em Mathematical methods in Physics: Distributions, Hilbert space
  operators, variational methods, and applications in quantum physics},
  volume~69.
\newblock Birkh{\"a}user, 2015.

\bibitem{connor1956historical}
M.B. Connor.
\newblock A historical survey of methods of solving cubic equations.
\newblock Master's thesis, University of Richmond, 1956.

\bibitem{descartes2012geometry}
R.~Descartes.
\newblock {\em The Geometry of Rene Descartes: With a facsimile of the first
  edition}.
\newblock Courier Corporation, 2012.

\bibitem{bruning2012parametrizations}
E.~Br{\"u}ning, H.~M{\"a}kel{\"a}, A.~Messina, and F.~Petruccione.
\newblock Parametrizations of density matrices.
\newblock {\em J. Mod. Opt.}, 59(1):1--20, 2012.

\bibitem{armitage2006elliptic}
J.V. Armitage and W.~F. Eberlein.
\newblock {\em Elliptic functions}, volume~67.
\newblock Cambridge University Press, 2006.

\end{thebibliography}
\end{document}